\documentclass[sigconf]{acmart}

\usepackage{etoolbox}
\usepackage{lineno}
\usepackage{hyperref}
\modulolinenumbers[5]
\usepackage[noend]{algpseudocode}
\usepackage[linesnumbered,ruled,noend,vlined]{algorithm2e}
\usepackage{caption}
\usepackage[singlelinecheck=on]{caption}
\usepackage{amsthm} 
\usepackage{colortbl}
\usepackage{booktabs}
\usepackage{float}
\usepackage[utf8]{inputenc} 
\usepackage[T1]{fontenc}
\usepackage{flushend}
\usepackage{graphicx}
\usepackage{mathrsfs}
\usepackage{multirow}
\usepackage{accents}
\usepackage{ragged2e}
\usepackage{setspace}
\usepackage{subfigure}
\usepackage{tabularx}
\usepackage{textcomp}
\usepackage{url}
\usepackage{verbatim}
\usepackage{xcolor}
\usepackage{xspace}
\usepackage{makecell}
\usepackage{dsfont}

\usepackage{subcaption}

\usepackage{amsmath,amssymb,amsfonts}

\newtheorem{example}{Example}

\newtheorem{theorem}{Theorem}
\newtheorem{lemma}{Lemma}

\newtheorem{definition}{Definition}

\AtBeginDocument{%
  \providecommand\BibTeX{{%
    \normalfont B\kern-0.5em{\scshape i\kern-0.25em b}\kern-0.8em\TeX}}}

\setcopyright{acmcopyright}
\copyrightyear{2025}
\acmYear{2025}
\acmDOI{10.1145/3698803}

\acmConference[SIGMOD '25]{}{June 22-27, 2025}{Berlin, Germany}

\acmBooktitle{Proceedings of the 2025 International Conference on Management of Data (SIGMOD ’25), June 22-27, 2025, Berlin, Germany}

\newcommand{\fnaive}{\widetilde{f}_1}
\newcommand{\foneround}{\widetilde{f}_2}

\newcommand{\errnaive}{O(\frac{ n_1^2 e^{4 \varepsilon}}{(1 + e^{\varepsilon})^4} )}

\newcommand{\erroneround}{O(\frac{ n_1 e^{2 \varepsilon}}{(1 - e^{\varepsilon})^4} )}

\newcommand{\fds}{f^*}
\newcommand{\pqx}{\mathcal{C}_2(u,w)}

\newcommand{\vq}{$u$\xspace}
\newcommand{\vx}{$w$\xspace}

\newcommand{\fq}{\widetilde{f}_u}
\newcommand{\fx}{\widetilde{f}_w}

\newcommand{\epldp}{$\varepsilon$-edge LDP\xspace}

\newcommand{\myspace}{\vspace{-0.2cm}}

\newcommand{\figsizeone}{0.23}
\newcommand{\figsizetwo}{0.46}

\newcommand{\cdp}{\texttt{CentralDP}\xspace}
\newcommand{\naive}{\texttt{Naive}\xspace}
\newcommand{\bs}{\texttt{OneR}\xspace}
\newcommand{\advss}{\texttt{MultiR-SS}\xspace} 
\newcommand{\advds}{\texttt{MultiR-DS}\xspace} 
\newcommand{\advdsbasic}{\texttt{MultiR-DS-Basic}\xspace} 
\newcommand{\advdeg}{\texttt{MultiR-DS*}\xspace} 
\newcommand{\advsslong}{Multiple-round Single Source\xspace} 
\newcommand{\advdslong}{Multiple-round Double Source\xspace} 

\newcommand{\figuresixa}{{Fig.~6(a)}\xspace}

\begin{document}

\title{Common Neighborhood Estimation over Bipartite Graphs under Local Differential Privacy}

\author{Yizhang He}
\affiliation{%
  \institution{University of New South Wales}
  \country{}
  }
\email{yizhang.he@unsw.edu.au}

\author{Kai Wang}
% \authornote{Kai Wang is the corresponding author.}
\affiliation{
  \institution{Antai College of Economics and Management, Shanghai Jiao Tong University}
  \country{}
  }
\email{w.kai@sjtu.edu.cn}

\author{Wenjie Zhang}
\affiliation{%
  \institution{University of New South Wales}
  \country{}
  }
\email{wenjie.zhang@unsw.edu.au}

\author{Xuemin Lin}
\affiliation{
  \institution{Antai College of Economics and Management, Shanghai Jiao Tong University}
  \country{}
  }
\email{xuemin.lin@sjtu.edu.cn}

\author{Ying Zhang}
\affiliation{%
  \institution{University of Technology Sydney}
  \country{}
  }
\email{ying.zhang@uts.edu.au}

\pagestyle{plain}

\begin{abstract}
\label{abs}
Bipartite graphs, formed by two vertex layers, arise as a natural fit for modeling the relationships between two groups of entities. In bipartite graphs, common neighborhood computation between two vertices on the same vertex layer is a basic operator, which is easily solvable in general settings. However, it inevitably involves releasing the neighborhood information of vertices, posing a significant privacy risk for users in real-world applications. To protect edge privacy in bipartite graphs, in this paper, we study the problem of estimating the number of common neighbors of two vertices on the same layer under edge local differential privacy (edge LDP). The problem is challenging in the context of edge LDP since each vertex on the opposite layer of the query vertices can potentially be a common neighbor. To obtain efficient and accurate estimates, we propose a multiple-round framework that significantly reduces the candidate pool of common neighbors and enables the query vertices to construct unbiased estimators locally. Furthermore, we improve data utility by incorporating the estimators built from the neighbors of both query vertices and devise privacy budget allocation optimizations. These improve the estimator's robustness and consistency, particularly against query vertices with imbalanced degrees. Extensive experiments on 15 datasets validate the effectiveness and efficiency of our proposed techniques.
\end{abstract}

\maketitle

% \pagestyle{\vldbpagestyle}
% \begingroup\small\noindent\raggedright\textbf{PVLDB Reference Format:}\\
% \vldbauthors. \vldbtitle. PVLDB, \vldbvolume(\vldbissue): \vldbpages, \vldbyear.\\
% \href{https://doi.org/\vldbdoi}{doi:\vldbdoi}
% \endgroup
% \begingroup
% \renewcommand\thefootnote{}\footnote{\noindent
% This work is licensed under the Creative Commons BY-NC-ND 4.0 International License. Visit \url{https://creativecommons.org/licenses/by-nc-nd/4.0/} to view a copy of this license. For any use beyond those covered by this license, obtain permission by emailing \href{mailto:info@vldb.org}{info@vldb.org}. Copyright is held by the owner/author(s). Publication rights licensed to the VLDB Endowment. \\
% \raggedright Proceedings of the VLDB Endowment, Vol. \vldbvolume, No. \vldbissue\ %
% ISSN 2150-8097. \\
% \href{https://doi.org/\vldbdoi}{doi:\vldbdoi} \\
% }\addtocounter{footnote}{-1}\endgroup

% \ifdefempty{\vldbavailabilityurl}{}{
% \vspace{.3cm}
% \begingroup\small\noindent\raggedright\textbf{PVLDB Artifact Availability:}\\
% The source code, data, and/or other artifacts have been made available at \url{\vldbavailabilityurl}.
% \endgroup

\section{Introduction}
Bipartite graphs have been widely used to represent connections between two sets of entities. 
Examples of real-world bipartite graphs include user-item networks in E-commerce \cite{wang2006unifying, li2020hierarchical}, people-location networks in contact tracing \cite{chen2021efficiently}, and user-page networks in social analysis \cite{o2010essentials}. 
\begin{figure}[thb]
\centering  
\includegraphics[width=0.45\textwidth]{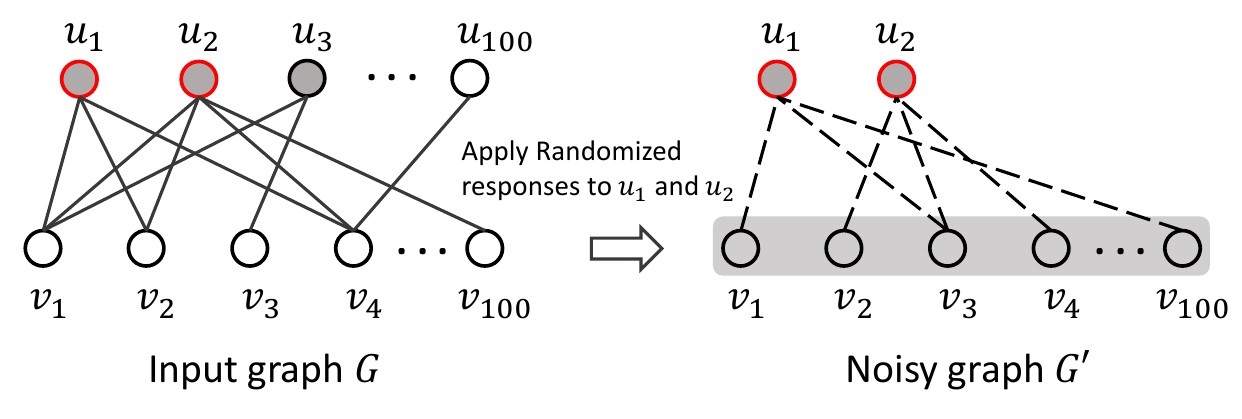}
\myspace
\caption{
A bipartite graph and its corresponding noisy graph by applying randomized responses to $u_1$ and $u_2$. 
}
\myspace
\myspace
\label{fig:motivation}
\end{figure}
%% backbone
{
In bipartite graphs, finding the common neighbors of two vertices is a basic operation in many tasks. For example, the similarity between two vertices can be computed using Jaccard similarity, which is the ratio of the number of their common neighbors to their combined unique neighbors \cite{leicht2006vertex, tsourakakis2014toward, yang2022efficient}. 
Additionally, common neighbor counts can help prune unpromising vertices in $(p,q)$-biclique counting \cite{yang2021p, ye2023efficient}. 
Other tasks that benefit from counting common neighbors in bipartite graphs include anomaly detection \cite{sun2005neighborhood}, bipartite graph projection \cite{stankova2021node, zhang2023bipartite}, bipartite clustering coefficient computation \cite{aksoy2017measuring, huang2010link}, community search \cite{dong2021butterfly, abidi2022searching, DBLP:journals/pvldb/WangWLZZ24}, 
and wedge-based motif counting \cite{xu2022efficient, wang2023efficient}. }

Although computing common neighbors is straightforward in the conventional setting, it inevitably involves releasing the neighborhood information of the vertices, posing a significant privacy risk for users in real-world applications. For instance, in user-item networks, disclosing identical items in the shopping carts of two users in online shopping platforms (e.g., eBay and Amazon) significantly compromises users' privacy. 
Hence, it is crucial to estimate the common neighborhood in a privacy-preserving manner, which remains an unresolved research gap.

%%% DP is the gold standard for privacy-preserving
In the literature, {\em differential privacy} (DP) \cite{dwork2014algorithmic} has become the gold standard for privacy-preserving computation, which provides a mathematical framework to quantify permissible privacy loss. 
%%% for graph data, edge LDP is widely used. 
Among the various DP models, {\em edge local differential privacy} (edge LDP) \cite{qin_generating_2017, zhang2018two, ye2020lf} has been widely adopted to protect the user's private connections. 
Edge LDP is a robust privacy protection protocol that requires each vertex to perturb its local data (e.g., degrees and neighbors) before transmitting it to the data curator. 

In this paper, we study the problem of privacy-preserving common neighborhood estimation over bipartite graphs. Specifically, given a bipartite graph $G$ and two query vertices \vq and \vx, we aim to estimate the number of common neighbors of \vq and \vx in $G$ on the same vertex layer with edge LDP. 
%% what is edge LDP
A random algorithm satisfies edge LDP when the probabilities of observing its output from any two vertices, whose neighbor sets differ by one vertex, are indistinguishable within a factor of $e^{\varepsilon}$. This ensures deniability for the existence of the edge $(u,v)$.
%%% what is privacy budget
In this context, $\varepsilon$ represents the privacy budget, determining the acceptable level of privacy loss. Clients can adjust this parameter to tradeoff between privacy and data utility. 
{
In addition, common neighbor counting under edge LDP is the first step in addressing other problems under edge LDP, such as vertex similarity computation \cite{leicht2006vertex} and $(p,q)$-biclique counting \cite{yang2023p}. 
}

\noindent
{\bf Challenges.} 
In this paper, we aim to design accurate and efficient common neighbor estimation algorithms on bipartite graphs with edge LDP. We face the following three main challenges. 

\noindent 
{\em $\bullet$ Challenge 1:} 
In the literature, Randomized Response \cite{warner1965randomized} is widely employed to achieve edge LDP, flipping the entries of the adjacency matrix between ``0'' and ``1'' to construct a noisy graph. However, counting common neighbors on the noisy graph results in severe overcounting and bias, because the noisy graph is generally much denser. For instance, we examine the performance of this naive approach (\naive) across 1000 runs for a pair of query vertices on the dataset \texttt{rmwiki} in Fig..~\ref{fig:distribution}. 
The blue distribution, representing the estimates of \naive, deviates significantly to the right from the true count (indicated by the black dashed line). This substantial shift highlights the difficulty in accurately estimating the number of common neighbors using the noisy graph.

% Some of the existing approaches \cite{imola2021locally, imola2022communication} adopt {\em Randomized Response} \cite{warner1965randomized} to provide edge LDP, which perturbs each vertex's edges and uses these noisy edges to construct a noisy graph. 
%% then briefly explain that RR will result in a very noisy graph 
% When constructing the noisy graph to satisfy edge LDP, randomized responses are applied to certain rows and columns of the adjacency matrix, which represents the connections between vertices in the bipartite graph. In the affected rows and columns, each entry is flipped (between ``0'' and ``1'') with a fixed probability. 
%% conclusion, this noisy graph presents challenges
% This process results in a much denser noisy graph, posing a challenge in accurately estimating common neighbors without bias.

\noindent 
{\em $\bullet$ Challenge 2:} 
Due to the constraints of edge LDP, we have to start with all vertices on the opposite layer of the query vertices as a candidate pool to estimate the common neighbors. 
This involves potentially $O(n)$ independent random variables, leading to a substantial margin of error. 
Reducing this candidate pool and developing an unbiased estimator that relies on fewer random variables for enhanced data utility presents a challenging task. 
% (1) the large candidate pool of common neighbors that includes all vertices on the opposite layer of the query vertices and 
% The candidate pool of the common neighbors includes all vertices on the opposite layer of the query vertices, leading to poor data utility. 
%%% challenge: how to construct an unbiased estimator that relies on fewer random variables, which is more likely to incur smaller variance. 

\noindent 
{\em $\bullet$ Challenge 3:} 
% \textcolor{red}{TO Do: elaborate}
% When the vertices and the data curator can interact for multiple rounds, determining the optimal way to distribute the privacy budget to different rounds based on the query vertices is difficult. 
% When the vertices and the data curator can interact for multiple rounds, the privacy budget needs to be split to be allocated to each round. The privacy budgets allocated to each round sum up to epsilon, meaning that if more is devoted to one round, less is available for others, making it challenging to manage errors through privacy budget distribution across multiple rounds.
% When vertices and the data curator interact over multiple rounds, the privacy budget must be divided and allocated to each round. 
% However, because the privacy budgets allocated to each round sum up to epsilon, allocating more to one round reduces the budget available for others, posing challenges in managing errors through privacy budget distribution across multiple rounds.
When the vertices and the data curator can interact for multiple rounds, the privacy budget must be divided among each round. However, allocating more budget to one round reduces the budget available for others. 
In addition, different pairs of query vertices likely need to be handled differently. 
% For instance, when the vertex degrees are highly imbalanced, how
Thus, finding the optimal allocation of privacy budgets to different rounds based on the query vertices requires special attention.

\noindent
{\bf Our approaches.} 
To address Challenge 1, we propose a one-round algorithm \bs that obtains unbiased estimates of common neighbors by leveraging the probability at which the entries in the adjacency matrix are flipped to compensate for over-counting. 
Specifically, for each query vertex, \bs applies randomized responses to both query vertices to generate noisy edges. 
For every vertex $v$ on the opposite layer of the query vertices, we estimate whether $v$ is a common neighbor of $u$ and $w$. 
Then \bs aggregates these estimates to obtain an unbiased estimate. 
While \bs addresses the challenge of the dense noisy graph and generates unbiased estimates, it relies on a large pool of candidates, compromising the utility of the data. 
As shown in Fig.~\ref{fig:distribution}, the yellow distribution representing the estimates of \bs appears symmetrical around the true common neighbor count with fat tails on both sides. 
This implies that the estimates are unbiased but have high variance.

To address Challenge 2, we propose a multiple round framework, allowing the query vertices to download the noisy edges from the previous round and reduce the candidate pool to their neighbors. 
Under this framework, we devise a single-source algorithm \advss, which returns an unbiased estimator for the number of common neighbors by leveraging the neighborhood of $u$. 
%% step 1. neighbor list not yet define, just say randomized response. 
Specifically, in the first round, \advss applies randomized responses to vertex $w$ to generate noisy edges. 
In the second round, the vertex $u$ retrieves these noisy edges from the data curator and then constructs a single-source estimator using its neighbors in the original graph and the neighbors of $w$ in the noisy graph. 
%%% the Laplacian step
To comply with edge LDP, the Laplace mechanism \cite{dwork2014algorithmic} is used to add noise to the single-source estimator before it is released to the data curator. 
This noise is scaled with the global sensitivity, defined as the maximum difference observed in the single-source estimator between two vertices whose neighbors differ by one edge. 
\advss achieves significantly better data utility compared to \bs by employing a multiple-round framework that reduces the candidate pool of common neighbors to the neighborhoods of the query vertices. 
As shown in Fig.~\ref{fig:distribution}, the green distribution represents the estimates of \advss, which preserves the unbiasedness and is more concentrated around the true value compared to \bs. 

\begin{figure}[thb]
\centering  
\includegraphics[width=0.38\textwidth]{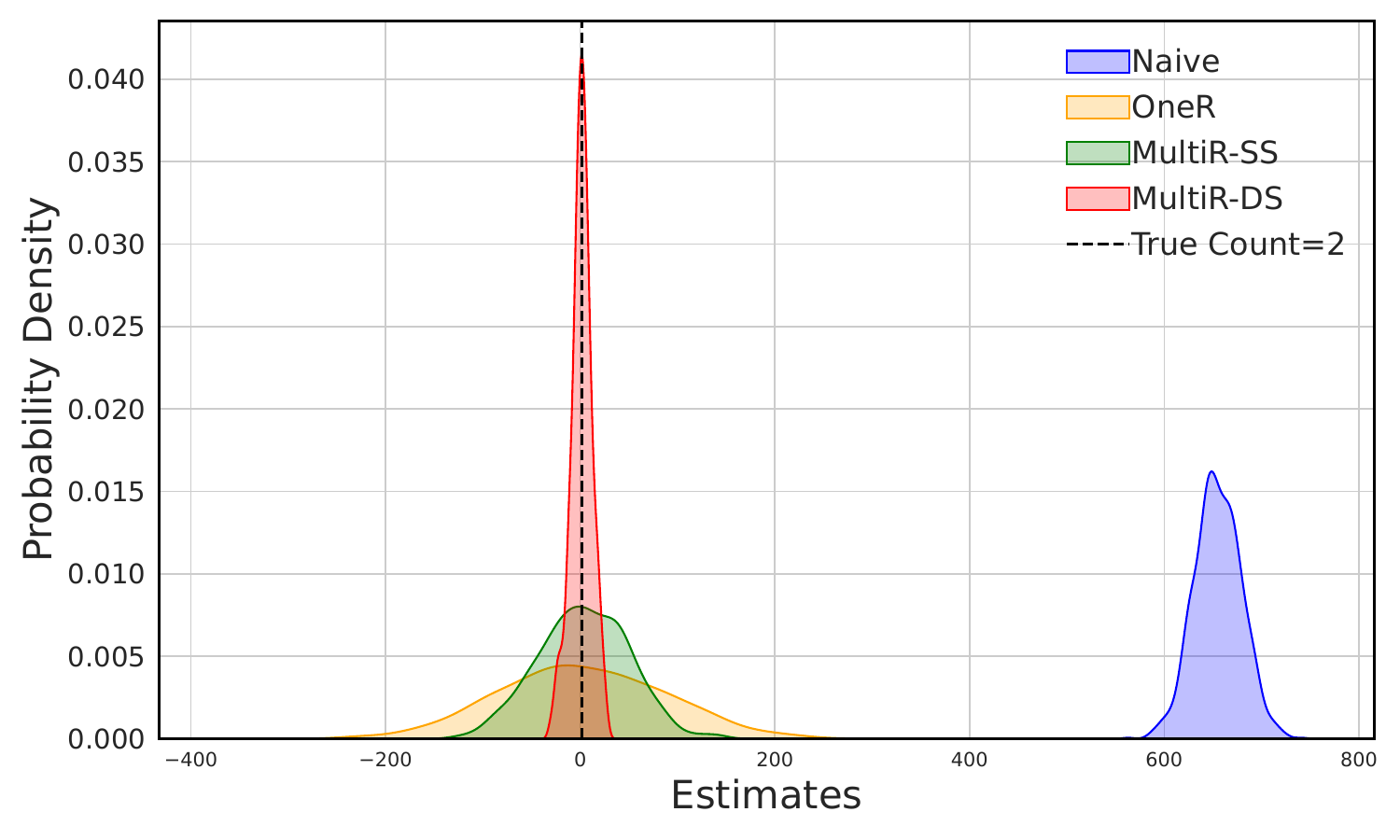}
\myspace
\caption{The estimate distribution on \texttt{rmwiki} when $\varepsilon=1$.
}
\myspace
\myspace
\label{fig:distribution}
\end{figure}

%% overview 
To tackle Challenge 3, we propose a double-source algorithm \advds under the multiple-round framework, which integrates two single-source estimators via a weighted average. 
In addition, we propose novel privacy budget allocation optimizations that allow \advds to dynamically adjust its privacy budgets for different rounds and the contribution of each single-source estimator for minimized L2 loss. 
%% steps
% By analyzing the L2 loss of the double-source estimator, we discover that it depends on the degrees of both query vertices, the privacy budget allocation to different rounds, and the weight of each single-source estimator. 
% Then, we utilize Newton's method to find the best privacy budgets for different rounds and the 
% impact 
% Thus, we utilize a small part of the privacy budget to estimate the query vertices' degrees and formulate the L2 loss as a multi-variable function, which can be minimized via Newton's method. 
%% describe the effects of these optimizations 
Specifically, when the incoming query vertices have large degrees, \advds tends to devote more privacy budget to noisy graph construction. %%% because?  
When the query vertices have very imbalanced degrees, \advds will favor the single-source estimator associated with the low-degree vertex, which depends on fewer random variables and induces less variance. 
%%% because the 
%%% we discover that the L2 loss of the double source estimator depends on the degrees of both query vertices, the privacy budget allocation, and the weight parameter. 
%%% Thus, we utilize a small proportion of the privacy budget to estimate the degrees of the query vertices and then leverage Newton's method to find epsilon 1 and alpha that minimize L2 loss. 
%%% specific steps (what it does)
% Specifically, in the second round, we build another unbiased estimator $\fx$ by switching $u$ and $w$. 
% Then, we build a double-source estimator $\fds$ by taking a weighted average of $\fq$ and $\fx$. 
%%% what it accomplishes
In doing so, \advds further reduces L2 loss compared to \advss and is more robust to query vertices with high degrees or unbalanced degrees. 
In Fig.~\ref{fig:distribution}, the red distribution depicts the estimates of \advds. 
Here the degrees of the query vertices are highly imbalanced ($556$ and $2$). 
In this case, \advds yields unbiased and more concentrated estimates compared to \advss. 

\noindent
{\bf Contributions.} Here we summarize our principal contributions. 

\noindent 
{\em $\bullet$} To the best of our knowledge, we are the first to study accurate and efficient common neighborhood estimation on bipartite graphs under edge LDP.

\noindent 
{\em $\bullet$} 
To address the over-counting issue with the \naive algorithm, we propose a one-round algorithm \bs to return unbiased estimates, which leverages the probabilistic nature of the noisy graph to compensate for over-counting. 

\noindent 
{\em $\bullet$} 
% We adopt the multiple-round framework and propose the multiple-round single-source algorithm, which xxxx. 
We propose a multiple-round framework and devise a single-source algorithm (\advss), 
which allows one query vertex to download the noisy edges from the other query vertex and construct an unbiased estimator locally. 
This drastically reduces L2 loss because the search scope for the common neighbors is reduced to the neighborhood of one query vertex.

\noindent 
{\em $\bullet$} 
Under the multiple-round framework, we propose a double-source algorithm (\advds) that constructs two unbiased estimators for each query vertex and combines them through a weighted average. \advds further reduces L2 loss by dynamically adjusting the allocation of the privacy budget and the contribution of two estimators based on the incoming query vertices.

\noindent 
{\em $\bullet$}
We conduct extensive experiments on $15$ real-world datasets to evaluate the proposed algorithms. The multiple-round algorithms \advss and \advds produce significantly smaller mean relative errors than \naive and \bs across all datasets. 
\advds is especially robust to query vertices with imbalanced degrees. 

\section{Preliminary}
\label{preliminary}
\begin{table}[htb]
\centering
\caption{Summary of Notations}
\myspace
\scalebox{0.9}{
    \begin{tabular}{c|c}
    \noalign{\hrule height 1pt}
    \cellcolor{gray!25} Notation & \cellcolor{gray!25} Definition \\ 
    \noalign{\hrule height 0.6pt}
    $G$ & a bipartite graph \\
    $\mathcal{A} $ & the adjacency matrix\\ 
    $\mathcal{A}_u$ & the neighbor list of the vertex $u$ \\  
    $N(u,G)$ & the neighbors of $u$ in $G$ \\
    $deg(u,G)$ & the degree of $u$ in $G$ \\
    $\varepsilon$ & a privacy budget \\ 
    $G'$ & a noisy bipartite graph\\
    % $\BTF_G$ & the number of butterflies in $G$ \\
    $\pqx$ & the number of common neighbors between $u$ and $w$\\
    $\Delta g$ & the global sensitivity of a function $g$\\
    % $\foneround$ & the one-round estimator\\
    $\fq$ & the single-source estimator based on $N(u, G)$\\
    $\fx$ & the single-source estimator based on $N(w, G)$\\
    $\fds$ & the double source estimator based on both $\fq$ and $\fx$\\
    \noalign{\hrule height 1pt}
    \end{tabular}
}
\label{tab:notation}
\myspace
\end{table}

\subsection{Problem definition}

We consider an unweighted bipartite graph $G(V=(U,L),E)$. $V=U\cup L$ denotes the set of vertices, where $U$ and $L$ represent the upper and lower layer, respectively. The vertices in $U$ and $L$ are called the upper vertices and the lower vertices. $E \subseteq U \times L$ denotes the set of edges. We use $n_1 = |U(G)|$ and $n_2 = |L(G)|$ to denote the number of upper and lower vertices, respectively, and $m$ = $|E|$ to represent the number of edges in $G$. The adjacency matrix $\mathcal{A}$ for $G$ is of dimensions $n \times n$, where $\mathcal{A}[u,v]=1$ if there exists an edge between the vertices $u$ and $v$ and $0$ otherwise. The $u$-th row of $\mathcal{A}$ (including both ``1'' and ``0'') is the {\em neighbor list} of $u$, denoted by $\mathcal{A}_u$. In addition, the set of neighbors of a vertex $u$ in $G$ is denoted by $N(u,G)$, and its degree is denoted by $deg(u,G)$ = $|N(u,G)|$. We use $d_{max}(U)$ and $d_{max}(L)$ to represent the maximum degree among the upper vertices and lower vertices, respectively. 
\begin{definition}
\label{def:pqx}
{\bf Common neighbors.} 
Let \vq and \vx be two vertices on the same layer of a bipartite graph $G$. 
The common neighbors of \vq and \vx are the vertices adjacent to both \vq and \vx in $G$, i.e., $N(u, G) \cap N(w, G)$. Here $N(u, G)$ represents the set of neighbors of vertex $u$ in graph $G$. 
We use $\pqx$ to denote the number of common neighbors of \vq and \vx, 
i.e., $\pqx = |N(u, G) \cap N(w, G)|$. 
\end{definition}

Classic differential privacy (DP) operates under a central model, where a central data curator manages a dataset $D$ \cite{dwork2014algorithmic}. $D$ and $D'$ are neighboring datasets if they differ by one data record. DP ensures that query outputs on these neighboring datasets are hard to distinguish. 
When extending DP under the central model to protect edge privacy in graphs, neighboring datasets refer to two graphs that vary by a single edge. 
However, the assumption with the central model that the data curator with access to the entire graph can be trusted can be impractical in real-world scenarios. 
Thus, we adopt the widely-used $\varepsilon$-edge local differential privacy (edge LDP), which enables each vertex to locally perturb its data before transmission to the data curator \cite{qin2017generating, zhang2018two, ye2020lf}. Under edge LDP, the neighboring datasets are two neighbor lists differing by one bit. 

\begin{definition} 
\label{def:ldp}
{\bf $\varepsilon$-edge local differential privacy.} 
Let $G$ be a bipartite graph and $\varepsilon>0$. 
For each vertex $u \in V(G)$, let $R_u$ with domain $\{0, 1\}^n$ be a randomized algorithm of vertex $u$. 
$R_u$ provides $\varepsilon$-edge LDP if for any two neighbor lists $\mathcal{A}_u$, $\mathcal{A}_u'$ that differ in one bit and any $S \subseteq Range(Ri)$: 
$$ Pr(R_u(\mathcal{A}_u) \in S  ) \leq e^{\varepsilon}  Pr(R_{u}(\mathcal{A}_u') \in S  ) $$
Here $\varepsilon$ is called the privacy budget. 
\end{definition}

\noindent
{\bf Problem Statement.}
Given a bipartite graph $G$, a privacy budget $\varepsilon$, and a pair of vertices \vq and \vx on the same layer of $G$, the {\em common neighborhood estimation} problem aims to estimate $\pqx$ while satisfying \epldp. 

{\color{black}Without loss of generality, we assume that $u$ and $w$ are in the lower layer of $G$.} In this paper, we use the {\em expected L2 loss} to evaluate the quality of estimates for the number of common neighbors. 
\begin{definition}
\label{def:l2}
{\bf Expected L2 loss.} % expected squared error
Given two vertices $u$ and $w$ on a bipartite graph $G$, and an estimate $f$ for the number of common neighbors between $u$ and $w$, the expected L2 loss of $f$ is the expected squared error between $\pqx$ and $f$, i.e., $$l2( \pqx, f ) = \mathbb{E}(( \pqx - f  )^2)$$
\end{definition}

\subsection{Warm up}
% Note that \epldp ensures that if two vertices have neighbor lists that differ by one bit, they cannot be reliably distinguished based on the outputs from the randomized algorithms. 
\epldp ensures that if two vertices have neighbor lists differing by just one bit, they cannot be reliably distinguished based on the outputs of the randomized algorithms. 
In this part, we introduce the most common methods for providing edge LDP, which are randomized responses and the Laplace mechanism. 
Then, we present a naive solution to the common neighborhood estimation problem.

\begin{figure}[thb]
\includegraphics[width=0.40\textwidth]{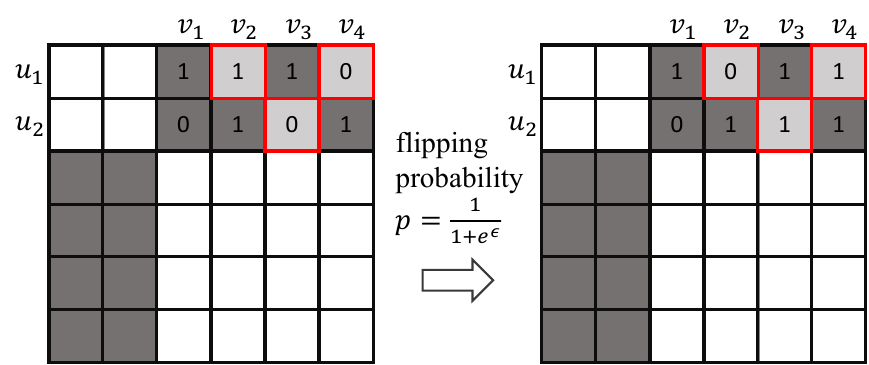}
\myspace
\caption{randomized responses on a bipartite graph. 
}
\label{fig:matrix}
\myspace
\end{figure}

\noindent
{\bf Warner's randomized response. }
One effective method for achieving \epldp is through randomized responses (RR), initially introduced as a survey technique to allow confidential answers to sensitive inquiries such as criminal or sexual activities \cite{warner1965randomized}. 
In essence, participants are asked to answer the questions honestly with probability. 
This concept has been adapted for graph applications to ensure edge local differential privacy \cite{qin_generating_2017, imola2021locally}. 
Specifically, each entry $x \in {0,1}$ of a neighbor list is perturbed with a probability of $\frac{1}{1+e^{\varepsilon}}$, where $\varepsilon$ denotes the privacy budget. 
$$
RR(x) = \begin{cases} 
1-x & \text{with probability } \frac{1}{1+e^{\varepsilon}} \\
x & \text{with probability } \frac{e^{\varepsilon}}{1+e^{\varepsilon}}
\end{cases}
$$
Current state-of-the-art methods on general graphs perturb each entry in the lower triangle of the adjacency matrix \cite{imola2021locally,imola2022communication}. 
However, to ensure that the resulting noisy graph $G'$ is bipartite, we only perturb entries in $\mathcal{A}$ that represent potential edges in a bipartite graph. 
Without loss of generality, we assume that the upper vertices have smaller IDs than the lower vertices. 
In this way, the adjacency matrix $\mathcal{A}$ can be divided into 4 blocks, where the ones on the diagonal are empty because the vertices of the same layer cannot be connected. 
%%%
In this paper, we only apply randomized responses to the neighbor lists of the query vertices \vq and \vx. 
When $u, w \in U(G)$, we flip the entries in $\mathcal{A}_u \cap L(G)$ and $\mathcal{A}_w \cap L(G)$. 
When $u, w \in L(G)$, we flip the entries in $\mathcal{A}_u \cap U(G)$ and $\mathcal{A}_w \cap U(G)$. 
In doing so, we avoid generating noisy edges not allowed in bipartite graphs. 
We denote the noisy graph with respect to a privacy budget $\varepsilon$ by $G'_{\varepsilon}$.
\begin{example}
We illustrate applying randomized responses to a bipartite graph with two upper vertices and four lower vertices in Fig.~\ref{fig:matrix}. 
The left shows the original adjacency matrix and the right side shows the matrix after applying randomized responses. 
Rows and columns are ordered with upper vertices preceding lower vertices. 
White squares indicate zeros, as edges within the same layer are not allowed in bipartite graphs. 
To find common neighbors of the upper vertices $u$ and $w$, we apply randomized responses to their neighbor lists, which affect the upper right block with eight grey squares. 
With the privacy budget \(\epsilon\), each grey block is flipped with probability \(\frac{1}{1 + e^{\epsilon}}\). Light grey squares outlined in red represent flipped entries.
\end{example}

\noindent
{\bf Calibrating noise with global sensitivity. }
To achieve \epldp, it is necessary to add noise to any data transmitted from a vertex to the data curator. The Laplace mechanism is used for this purpose, which calibrates the amount of noise with the {\em global sensitivity} of the transmitted data. 

\begin{definition}
\label{def:gs}
{\bf Global sensitivity.} 
Consider a bipartite graph $G$. 
Let $\mathcal{A}_u$ be the neighbor list of vertex $u$. 
Let $\mathcal{A}_u'$ be a neighbor list that differs from $\mathcal{A}_u$ in at most one entry. 
The global sensitivity of a function $f: \mathcal{A}_u \rightarrow \mathbb{R}$ is:
$$
\Delta f = max_{\mathcal{A}_u , \mathcal{A}_u'} | f(\mathcal{A}_u) - f(\mathcal{A}_u') |
$$
% Here $\mathcal{A}_u'$ is a neighbor list that differs from $\mathcal{A}_u$ in at most one entry. 
% $\Delta f$ represents the global sensitivity of $f$. 
\end{definition}
\begin{definition}
\label{def:lap}
{\bf The Laplace Mechanism.} 
Given a privacy budget $\varepsilon$ and any function $f$, the Laplace mechanism is defined as:
$$
\tilde{f} = f + \text{Lap}\left(\frac{\Delta f}{\varepsilon}\right)
$$
Here $\tilde{f}$ is the noisy version of $f$ and $Lap(\cdot)$ is the probability density function of the Laplace distribution. 
\end{definition}
By applying the Laplace mechanism, we allow the vertices to send local graph statistics with calibrated noise to the data curator while satisfying \epldp. 
% Note that edge LDP is immune to {\em post-processing}, allowing the data curator to apply any post-processing to the graph statistics it receives from the vertices without compromising the privacy guarantees \cite{imola2021locally}.

{\color{black}
Edge LDP can also be obtained via composition (i.e., combine multiple edge LDP algorithms). 
For instance, {\em sequential composition} \cite{jiang2021applications, qin_generating_2017} enables the sequential application of multiple edge LDP algorithms ($M_i$), each consuming some privacy budgets ($\varepsilon_i$), and ensures that the overall process satisfies $\sum_{i} \varepsilon_i$-edge LDP. 
{\em Parallel composition} states that if disjoint subsets of the graph are processed by different edge LDP algorithms ($M_i$) with privacy budget $\varepsilon_i$, then the overall mechanism running these algorithms in parallel satisfies $\max_i(\varepsilon_i)$-edge LDP \cite{yang2023local}.
% allows disjoint subsets of the graph to be processed independently by different mechanisms, enabling each mechanism to utilize the full privacy budget $\epsilon$ \cite{yang2023local}. 
%% for instance. 
Furthermore, edge LDP is immune to {\em post-processing} \cite{yang2023local, imola2021locally}, allowing the data curator to apply any post-processing to the graph statistics received from the vertices without compromising the privacy guarantees. }

\begin{algorithm}[tbh]
    \small
	\caption{\naive}
	\label{algo:naive}
	\LinesNumbered
	\KwIn{$G$: a bipartite graph; 
            $\varepsilon$: a privacy budget; 
            \vq, \vx: two query vertices from the same layer} 
	\KwOut{$\fnaive$: the naive estimator of $\pqx$}
         \tcp{\textbf{vertex side:}} 
        %% this is for multiple round 
\ForEach{$i \in \{u, w\}$}{
    \ForEach{$j \in$ the opposite layer from $u$ and $w$}{
    perturb 
        $\mathcal{A}'[i,j]  \gets
        \begin{cases}
            1- \mathcal{A}[i,j] , &  \text{w.p. }  \frac{1}{1+ e^{\varepsilon}} \\
            \mathcal{A}[i,j] & \text{w.p. }  \frac{e^{\varepsilon}}{1+ e^{\varepsilon}}
        \end{cases}
        $
    } 
}
send noisy edges to the data curator;\\
$G_{\varepsilon}' \gets $ the noisy graph constructed from $\mathcal{A}'[i,j]$;\\
        \tcp{\textbf{curator side:}} 
        $\fnaive \gets |N(u, G_{\varepsilon}') \cap N(w, G_{\varepsilon}')|$;\\
        \textbf{return} $\fnaive$;\\
\end{algorithm}

\noindent
{\bf A naive approach.} 
Given that randomized responses preserve \epldp, a naive approach is to count the number of common neighbors of \vq and \vx on the noisy graph constructed by randomized responses, as outlined in Algorithm \ref{algo:naive}. 
Note that for our problem, we only need to apply randomized responses to \vq and \vx. 
Specifically, given a privacy budget $\varepsilon>0$, \naive flips the entries $A[i,j]$ with a probability $\frac{1}{1+e^{\varepsilon}}$, where $i \in \{u, w\}$ (Lines 1-4). 
The data curator collects the noisy edges from \vq and \vx and constructs a noisy graph $G_{\varepsilon}'$ (Line 5). 
\textcolor{black}{In this way, we do not need to analyze the global sensitivity for \naive because it does not involve the Laplace mechanism and only relies on randomized responses to provide edge LDP. }
Then, the naive estimator $\fnaive$ is calculated taking the intersection of the neighbors of \vq and \vx in $G_{\varepsilon}'$, i.e., $\fnaive = |N(u, G_{\varepsilon}') \cap N(w, G_{\varepsilon}')|$ (Line 6). 
However, since there are generally more ``0''s than ``1''s in the neighbor lists of \vq and \vx, applying randomized responses usually results in a much denser noisy graph $G_{\varepsilon}'$, which results in severe overcounting $\pqx$. 

\noindent
{\bf Theoretical analysis for \naive.} 
{\color{black}
Without loss of generality, we assume that the query vertices $u$ and $w$ are from the lower layer $L(G)$ when analyzing the time costs, communication costs, and expected L2 losses. 
%% time costs. 
The time costs are divided between the vertex side and the curator side.
On the vertex side, the time costs incurred by randomized responses is $O(n_1)$, where $n_1$ is the number of vertices in $U(G)$. 
On the curator side, the dominating cost is incurred by intersecting $N(u, G_{\varepsilon}')$ and $N(w, G_{\varepsilon}')$, which takes $O(min(deg(u,  G_{\varepsilon}'), deg(w,  G_{\varepsilon}')))$. Thus, the overall time complexity is $O(n_1)$. 
The communication costs are incurred only during randomized responses, where vertices $u$ and $w$ send noisy edges to the data curator. 
For vertex $u$, the expected number of noisy edges is $d_u \times (1-p) + (n_1 - d_u) \times p$, where $p = \frac{1}{1+e^{\varepsilon}}$. Similarly, for vertex $w$, the expected number of noisy edges is $d_w \times (1-p) + (n_1 - d_w) \times p$. 
Thus, the overall communication cost is $O\left(\frac{e^{\varepsilon}-1}{e^{\varepsilon}+1}(d_u + d_w) + \frac{2n_1}{1 + e^{\varepsilon}}\right)$.} 
In the following, we analyze the expected L2 loss of the estimator $\fnaive$ returned by \naive. 

\begin{theorem} 
\label{thm:f1}
{\color{black}
Given a bipartite graph $G$, a privacy budget $\varepsilon$, and a pair of query vertices \vq and \vx in $L(G)$, the expected L2-loss of the estimator for $\fnaive$ in Algorithm \ref{algo:naive} is $O(\frac{n_1^2}{(1+e^{\varepsilon})^4 }) $. 
Here, $n_1$ represents the number of vertices in $U(G)$. 
}

% Here $n$ represents the number of vertices in $G$. 
\end{theorem} 
\begin{proof}
{\color{black}
    The naive estimator $\fnaive = |N(u, G_{\varepsilon}')  \cap N(w, G_{\varepsilon}')|$. 
    Let $p = \frac{1}{1 + e^\varepsilon}$ be the flipping probability during the randomized responses. 
    Note that each entry $\mathcal{A}'[i,j]$ on the adjacency matrix of the noisy graph follows a Bernoulli distribution with a parameter $p$ (if $\mathcal{A}[i,j]=0$) or $1-p$ (if $\mathcal{A}[i,j]=1$). 
    Since $\fnaive\geq 0$, we have the following inequality:
        \begin{align*}
        & l2(\fnaive, \pqx)  
        % = \mathbb{E}( (\fnaive - \pqx)^2 )\\
        % & = \mathbb{E}( \fnaive^2 - 2\fnaive \pqx + \pqx^2  )  \\ 
        =  \mathbb{E}( \fnaive^2) -2 \pqx  \mathbb{E}(\fnaive) + \pqx^2
        % & \leq \mathbb{E}( \fnaive^2) + \pqx^2
        = O(\mathbb{E}( \fnaive^2)) 
        \end{align*}
        
Since $\fnaive = \sum_{v \in U(G) } \mathcal{A}'[u,v] \mathcal{A}'[v,w]$, we have: 
\begin{align*}
& \mathbb{E}( \fnaive^2) 
% = \mathbb{E}( (  \sum_{v \in U(G) } \mathcal{A}'[u,v] \mathcal{A}'[v,w]  )^2)\\
= \sum_{v \in U(G)} \mathbb{E}((\mathcal{A}'[u,v] \mathcal{A}'[v,w])^2) \\
& + 2 \sum_{v_i < v_j \in U(G)} \mathbb{E}( \mathcal{A}'[u,v_i] \mathcal{A}'[v_i,w]\mathcal{A}'[u,v_j] \mathcal{A}'[v_j,w]) \\ 
& \leq O({n_1 \choose 2 } (1-p)^4    ) = O(n_1^2 (1-p)^4 )= \errnaive
\end{align*}
% Thus, 
% \begin{align*}
% & l2(\fnaive, \pqx)  \leq \mathbb{E}( \fnaive^2) + \pqx^2 \\ 
% &= {\color{black}\mathbb{E}( (p(2p-3) \pqx + p(1-2p)|N(u, G)  \cup N(u, G)| + p^2 n_1)^2 )} \\
% & \leq 
% {\color{black}O((1-p)^4 n_1^2) = \errnaive }
% \end{align*}
The last step is due to $p = \frac{1}{1 + e^{\varepsilon}}$. 
}
% \begin{align*}
% & l2(\fnaive, \pqx)  = \mathbb{E}( (\fnaive - \pqx)^2 )\\
% &= {\color{black}\mathbb{E}( (p(2p-3) \pqx + p(1-2p)|N(u, G)  \cup N(u, G)| + p^2 n_1)^2 )} \\
% & \leq 
% {\color{black}O((1-p)^4 n_1^2) = \errnaive }
% \end{align*}
% The last step is due to $p = \frac{1}{1 + e^{\varepsilon}}$. 
\end{proof}

\begin{theorem}
{\color{black}
    The \naive algorithm satisfies \epldp. }
    % \textcolor{red}{TODO: provide formal proof. }
\end{theorem}
% \begin{proof}
% % The theorem follows because the randomized responses provide \epldp, which is immune to post-processing. 
% The randomized responses provide \epldp. Since edge LDP is immune to post-processing, the theorem holds. 
% % Since the randomized responses provide \epldp which is immune to post-processing, the theorem holds. 
% \end{proof}
\begin{proof}
{\color{black}
    Since the randomized responses provide \epldp \cite{imola2021locally, imola2022differentially}, 
    Lines 1-4 of the algorithm satisfy \epldp. 
    In addition, edge LDP is immune to post-processing, which means that any analysis (Lines 5-6) conducted on the noisy graph preserves edge LDP. 
    Thus, the theorem holds. 
}
\end{proof} 

\section{A one-round approach}
\label{sec:bs}

To address the problem of overcounting with the \naive approach, we propose a one-round algorithm \bs, which exploits the probabilistic nature of the noisy graph to obtain unbiased estimates of $\pqx$. 
Specifically, \bs leverages the flipping probability during randomized responses to construct an unbiased estimator for each vertex on the opposite layer of the query vertices. 
\text{In doing so, we do not need to analyze the global sensitivity for \bs, as it only relies on randomized responses to ensure edge LDP and does not involve the Laplace mechanism.} 
Then, \bs aggregates these estimates to obtain unbiased counts of common neighbors. 
% compute the probabilities at which each vertex on the opposite layer from the query vertices is a common neighbor. 
% the probability at which the entries on the adjacency matrix are flipped during randomized responses and estimates the likelihood of each vertex being a common neighbor of $u$ and $w$. 
% To get started, we first 
% To address Challenge 1, we propose a one-round algorithm \bs that obtains unbiased estimates of common neighbors by leveraging the probability at which the entries in the adjacency matrix are flipped to compensate for over-counting. 
First, we investigate the unbiased estimator of $\mathcal{A}[i,j]$ for two vertices $i$ and $j$ in the bipartite graph. 

% Note that for any two vertices $u$ and $v$, we can estimate the probability that the edge $(u,v)$ exists in the original graph. This also allows us to compute the probability for a given vertex to be a common neighbor of $u$ and $w$ by taking the product of the related edge existence probabilities. 

%%% an unbiased estimtor for A[i,j]
\subsection{An unbiased estimator for $\mathcal{A}[i,j]$}
Consider a bipartite graph $G$. Let $\varepsilon$ be the privacy budget. We use $G_{\varepsilon}'$ to represent the noisy graph from applying randomized responses to the edges in $G$. 
During randomized responses, each entry in the neighbor list is flipped with a probability $p  = \frac{1}{1+e^{\varepsilon}}$. 
% If $\mathcal{A}[i,j]$
Note that each entry $\mathcal{A}'[i,j]$ on the adjacency matrix of the noisy graph follows a Bernoulli distribution with a parameter $p$ (when $\mathcal{A}[i,j]=0$) or $1-p$ (when  $\mathcal{A}[i,j]=1$). 
% Note that the noisy edge $\mathcal{A}[i,j]'$ follows a Bernoulli distribution of $Bin(1-p)$ when $\mathcal{A}[i,j]=1$ and $Bin(p)$ when $\mathcal{A}[i,j]=0$. 
Based on this, we have the following equations which link the expected value of $\mathcal{A}[i,j]'$ and $\mathcal{A}[i,j]$: 
$$
\mathbb{E}(\mathcal{A}'[i,j] ) =
\begin{cases}
    p, & \text{if } \mathcal{A}[i,j] = 0 \\
    1-p & \text{if } \mathcal{A}[i,j]  = 1
\end{cases}
$$
This can be rearranged into one single equation: 
$
\mathbb{E}(\mathcal{A}'[i,j] ) =  \mathcal{A}[i,j] + p (1- 2 \mathcal{A}[i,j])$. 
Solving this equation for $\mathcal{A}[i,j]$ leads to: 
$$
E(\frac{\mathcal{A}'[i,j]  - p }{1-2p}) = \mathcal{A}[i,j]
$$
Let $\phi(i,j) = \frac{\mathcal{A}'[i,j]  - p }{1-2p}$. It immediately follows that $\phi(i,j)$ is an unbiased estimator of $\mathcal{A}[i,j]$. We can also analyze the variance of $\phi(i,j)$: 
\begin{align}
    \mathrm{Var}(\phi(i,j)) =  \mathrm{Var}(\frac{\mathcal{A}'[i,j]  - p }{1-2p}) 
    % = \frac{\mathrm{Var}(\mathcal{A}'[i,j])}{(1-2p)^2} \\ 
    = \frac{p(1-p)}{(1-2p)^2}   
    \label{eq:phi.var}
\end{align}
The last step is because $\mathcal{A}'[i,j]$ is a Bernoulli variable with a probability $p$ or $1-p$, which leads to a variance of $p(1-p)$. 
\begin{algorithm}[tbh]
    \small
	\caption{\bs}
	\label{algo:bs}
	\LinesNumbered
	\KwIn{$G$: a bipartite graph; 
            $\varepsilon$: a privacy budget; 
            \vq, \vx: two query vertices from the same layer} 
	\KwOut{$\foneround$: the one-round unbiased estimator for $\pqx$}
        %% this is for multiple round 
\ForEach{$i \in \{u, w\}$}{
    \ForEach{$j \in$ the opposite layer from $u$ and $w$}{
    perturb 
        $\mathcal{A}'[i,j]  \gets
        \begin{cases}
            1- \mathcal{A}[i,j] , &  \text{w.p. }  \frac{1}{1+ e^{\varepsilon}} \\
            \mathcal{A}[i,j] & \text{w.p. }  \frac{e^{\varepsilon}}{1+ e^{\varepsilon}}
        \end{cases}
        $
    } 
}
send noisy edges to the data curator;\\
$G_{\varepsilon}' \gets $ the noisy graph constructed from $\mathcal{A}'[i,j]$;\\
        $\foneround \gets 0$; $p \gets \frac{1}{1 + e ^{\varepsilon}}$;\\
        \ForEach{{\color{black}$v \in V(G'_{\varepsilon})$ on the opposite layer as \vq and \vx }}{
            $\foneround \gets \foneround + (\mathcal{A}'[u, v] - p) (\mathcal{A}'[v, w] - p) / (1-2p)^2  $;\\
       }
        \textbf{return} $\foneround$;\\
\end{algorithm}

\subsection{An unbiased estimator for $\pqx$}
In this part, we derive an unbiased estimator for $\pqx$ based on the noisy graph from randomized responses. 
{\color{black}Without loss of generality, let's assume that both \vq and \vx are in $L(G)$. 
By definition, $\pqx = | N(u, G) \cap N(w, G)  | = 
\sum_{v \in U(G) } \mathcal{A}[u,v] \mathcal{A}[v,w]$. 
This implies that we need to estimate $\mathcal{A}[u,v] \mathcal{A}[v,w]$ for all $v \in U(G)$ in an unbiased way. 
Since $\mathcal{A}'[u,v]$ and $\mathcal{A}'[v,w]$ are independent of each other, we have $\mathbb{E}(\phi(u, v) \phi(v, w)) = \mathcal{A}[u,v] \mathcal{A}[v,w]$, which leads to the following estimator for $\pqx$. }

\begin{theorem}
\label{thm:f2}
Consider a bipartite graph $G$, a privacy budget $\varepsilon$, and two vertices \vq and \vx in $L(G)$. 
Let $G_{\varepsilon}'$ be the noisy graph after applying the randomized response to $G$ w.r.t. $\varepsilon$. 
Let $p = \frac{1}{1+e^{\varepsilon}}$ be the flipping probability. 
We have
$
\mathbb{E}( \foneround(u, w) ) = \pqx
$
where  
\begin{align}
    \foneround(u, w) = \sum_{ v \in U(G)}  \frac{(\mathcal{A}'[u,v] - p) (\mathcal{A}'[v,w] - p )}{(1-2p)^2}
    \label{eq:fbs}
\end{align}
is an unbiased estimate for $\pqx$. 
\end{theorem}

The proof of Theorem \ref{thm:f2} involves summing $\phi(u, v) \phi(v, w)$ over all $v \in U(G)$. 
In particular, when \vq and \vx belong to $U(G)$, similar results are derived with $v$ in $L(G)$. 
Based on Theorem \ref{thm:f2}, we design a one-round algorithm \bs, which estimates $\pqx$ based on the noisy graph generated from randomized responses, as outlined in Algorithm \ref{algo:bs}. In Lines 1-6, \bs constructs a noisy graph $G'_{\varepsilon}$ by applying randomized responses to \vq and \vx. 
%%%
Then, it computes $\foneround(u, w)$ by considering all vertices on the opposite layer (i.e., $U(G)$) as candidates for the common neighbors of \vq and \vx (Lines 7-8). 
%% note that this step can be implemented efficiency by computing the intersection and union of $N(u, G'_{\varepsilon})$ and $N(w, G'_{\varepsilon})$. 
In practice, to efficiently compute $\foneround(u, w)$ and avoid visiting all candidate vertices, the expression in Equation \ref{eq:fbs} can be expanded as such. 
% Note that the expression in Equation \ref{eq:fbs} for $\foneround(u, w)$ can be expanded in the following for efficient computation. 
{\color{black}
\begin{align*}
\foneround(u, w) & = N_1  \frac{{(1 - p)^2}}{{(1 - 2p)^2}} 
\hspace{-0.5mm} - (N_2 - N_1)  \frac{{(1 - p)  p}}{{(1 - 2p)^2}} 
\hspace{-0.5mm} + ( n_1-N_2 )  \frac{{p^2}}{{(1 - 2p)^2}}
\end{align*}
}
Here $N_1$ denotes the number of common neighbors of $u$ and $w$ in the noisy graph $G_{\varepsilon}'$. 
$N_2$ denotes size of the union of the neighbor sets of $u$ and $w$ in $G_{\varepsilon}'$. 
{\color{black}$n_1$ represents the number of vertices in $|U(G)|$ (i.e., the opposite layer of $u$ and $w$). }
In this way, we only need to compute the intersection and union of the neighbor sets of $u$ and $w$ in $G_{\varepsilon}'$ to efficiently obtain $\foneround(u, w)$.

\begin{example}
% \textcolor{red}{TODO: need to mention here we use $n_2$}
We illustrate the construction of the unbiased estimator $\foneround$ returned by the \bs algorithm using the bipartite graph shown in Fig.\ref{fig:motivation}, focusing on $u_1$ and $u_2$ as query vertices with three common neighbors. 
The outline of the query vertices is highlighted in red. 
On the right side of Fig.\ref{fig:motivation}, we present the noisy graph constructed by applying randomized responses to $u_1$ and $u_2$. Note that we only need to apply randomized responses to $u_1$ and $u_2$. The dashed lines in the graph represent the resulting noisy edges. 
The vertices shaded in grey depict the candidate pool for common neighbors between $u_1$ and $u_2$ in \bs, including all vertices on the opposite layer from $u_1$ and $u_2$. 
According to Equation \ref{eq:fbs}, $\foneround$ relies on $n_2 = |L(G)|=100$ random variables: $\mathcal{A}'[v_i, u_2]$, where $i \in [1, 100]$. 
\end{example}

\noindent
{\bf Theoretical analysis for \bs.} %% time complexity and expected L2 loss
{\color{black}
Without loss of generality, we assume that the query vertices $u$ and $w$ are in $L(G)$ in the following analyses. 
On the vertex side, the time costs incurred by randomized responses is $O(n_1)$, where $n_1$ is the number of vertices in $U(G)$. 
On the curator side, the dominating cost is incurred by computing $\foneround(u, w)$ in Lines 7-8, which can be implemented in $O(deg(u, G'_{\varepsilon}) + deg(w, G'_{\varepsilon}))$ time by computing the intersection and the union of $N(u, G'_{\varepsilon})$ and $N(w, G'_{\varepsilon})$. The overall time complexity is $O(n_1)$. 
The communication costs of \bs are incurred only during randomized responses, similar to \naive. The overall communication cost is $O\left(\frac{e^{\varepsilon}-1}{e^{\varepsilon}+1}(d_u + d_w) + \frac{2n_1}{1 + e^{\varepsilon}}\right)$. }

In the following, we analyze the expected L2 loss of $\foneround$ returned by Algorithm \ref{algo:bs} in Theorem \ref{thm:f2.var}. 
% The steps on the vertex sides are the same as in \naive. The difference is on the data curator side, where in Lines xxxx the flipping probability $p$ is involved to correct for unbiased estimates. 

\begin{theorem}
\label{thm:f2.var}
{\color{black}
The L2 loss of $\foneround$ in Theorem \ref{thm:f2} is $O(\frac{ n_1  e^{\varepsilon}}{(1 - e^{\varepsilon})^4} )$. Here, $n_1$ represents the number of vertices in $U(G)$ (i.e., the opposite layer to \vq and \vx). }
\end{theorem}
\begin{proof}

Since $\foneround$ is unbiased, based on the bias-variance decomposition theorem \cite{bouckaert2008practical}, the L2 loss of $f_1$ equals its variance. 
{\color{black}Assume \vq and $w \in L(G)$. }
\begin{align*}
    & l_2( \foneround,  \pqx ) = \mathrm{Var}( \foneround(u,w))  \\
    % & = \hspace{-2mm} 
    % {\color{black}\sum_{v \in U(G)} \mathrm{Var}( \frac{(\mathcal{A}'[u,v] - p) (\mathcal{A}'[v,w] - p )}{(1-2p)^2} ) }\\
    & = {\color{black}\frac{1}{(1-2p)^4} \sum_{   v \in U(G)}  \mathrm{Var}( (\mathcal{A}'[u,v] - p) (\mathcal{A}'[v,w] - p ) ) }\\
\end{align*}
% Let $\eta =\mathcal{A}'[u,v] - p $ and $\theta = \mathcal{A}'[v,w] - p $. 
% We know $\mathbb{E}(\eta )  = 1-2p$ when $v \in N(u, G)$ and $\mathbb{E}(\eta )  = 0$ otherwise. 
% $\mathbb{E}(\theta )  = 1-2p$ when $v \in N(w, G)$ and $\mathbb{E}(\theta )  = 0$ otherwise. 
% This is because $\eta$ and $\theta$ are shifted Bernoulli variables with the same variances. 
% Also, we know $\mathrm{Var}(\eta ) = \mathrm{Var}(\theta) = p(1-p)$ because they are shifted Bernoulli variables with the same variances. 
{\color{black}Let \(\eta = \mathcal{A}'[u,v] - p\) and \(\theta = \mathcal{A}'[v,w] - p\). 
By construction, we know that \(\mathbb{E}(\eta) = 1 - 2p\) when \(v \in N(u, G)\) and \(\mathbb{E}(\eta) = 0\) otherwise. 
Similarly, \(\mathbb{E}(\theta) = 1 - 2p\) when \(v \in N(w, G)\) and \(\mathbb{E}(\theta) = 0\) otherwise. 
In addition, we have \(\mathrm{Var}(\eta) = \mathrm{Var}(\theta) = p(1 - p)\). 
This is because \(\eta\) and \(\theta\) are shifted Bernoulli variables with the same variance. }

\begin{align*}
& \mathrm{Var}( \foneround(u,w))  = \frac{1}{(1-2p)^4} \sum_{v\in U(G)}  \mathrm{Var}( \eta \theta ) \\
% &= \frac{1}{(1-2p)^4} \sum_{ v \in U(G)} (\mathrm{Var}(\eta)\mathrm{Var}(\theta) + \mathrm{Var}(\eta) \mathbb{E}(\theta)^2 + \mathrm{Var}(\theta) \mathbb{E}(\eta)^2   ) \\
% &= {\color{black} \frac{1}{(1-2p)^4} \sum_{  v \in U(G)}   p^2(1-p)^2 + p(1-p) ( \mathbb{E}(\theta)^2 +\mathbb{E}(\eta)^2 ) } \\ 
% &= {\color{red}\frac{1}{(1-2p)^4} \sum_{  v \in N(u) \cap N(w)} +  
%    \frac{1}{(1-2p)^4} \sum_{others}} \\ 
% & = {\color{red}\frac{1}{(1-2p)^4} \sum_{v \in N(u) \cap N(w) } p^2(1-p)^2 + 2 p(1-p) (1-2p)^2  + }\\
% &  {\color{red}\frac{1}{(1-2p)^4} \sum_{ otherwise} p^2(1-p)^2 +  p(1-p) (1-2p)^2 }\\
    % & = \frac{p^2(1-p)^2}{(1-2p)^4} |U(G)| + \frac{p(1-p)}{(1-2p)^2}(|U(G)|+\pqx) \\ 
&= {\color{black}\frac{p^2(1-p)^2}{(1-2p)^4} |U(G)|
    + \frac{p(1-p)}{(1-2p)^2}(deg(u, G) + deg(w, G)) } \\
&= {\color{black} O( \frac{p^2(1-p)^2}{(1-2p)^4}  |U(G)|  ) = \erroneround } 
\end{align*}
% \erroneround
The last step is due to $p = \frac{1}{1 + e^{\varepsilon}}$. 
\end{proof}

\begin{theorem}
{\color{black}
The \bs algorithm satisfies \epldp. 
}
\end{theorem}
\begin{proof}
% The theorem follows because the randomized responses provide \epldp, which is immune to post-processing. 
% The randomized responses provide \epldp. Since edge LDP is immune to post-processing, the theorem follows. 
{\color{black}
Since the randomized responses provide \epldp\cite{imola2021locally, imola2022differentially}, Lines 1-4 of the algorithm satisfy \epldp. 
Additionally, edge LDP is immune to post-processing, meaning that any analysis conducted on the noisy graph (Lines 5-8) preserves edge LDP. Thus, the theorem holds.
}
\end{proof}

\section{Multiple-round approaches}
\label{sec:adv}

% \textcolor{red}{TODO: add some illustrations}

% \textcolor{red}{Do we mention the multiple round framework}
The \bs algorithm reflects our first attempt at obtaining an unbiased estimate of $\pqx$. 
{However, as analyzed in Section \ref{sec:bs}, the L2 loss of \bs still contains a factor of $n_1$}, because \bs inevitably considers all potential vertices on the opposite layer from \vq and \vx as candidates for the common neighbors. 
% two hop paths $<u, v, w>$ where $v$ is any middle vertex on the opposite layer from \vq and \vx. 
% In the literature of motif counting under edge LDP, a classic approach is to adopt a multiple-round computing scheme, where we use one round to construct a noisy graph and we let the vertices combine their local neighborhood with the noisy graph and report useful graph statistics to be aggregated by the central data curator. 
To further improve data utility, in this section, we employ the classic multiple-round framework in the literature of graph analysis under edge LDP \cite{imola2021locally}. 
In the first round, we utilize a part of the privacy budget to construct a noisy graph via randomized responses. 
Then, we allow both \vq and \vx to download the noisy edges from each other and combine their local neighbors with the noisy edges to compute unbiased estimates of $\pqx$ locally. 
In the end, we use the remaining privacy budget to apply the Laplace mechanism to these unbiased estimates to comply with edge LDP. 
% where one round is dedicated to the noisy graph construction via randomized responses, after which the vertices combine their local neighborhoods with the noisy graph and report graph statistics with small variances. 
% In the literature on mot counting under edge LDP, a classic approach involves employing a multi-round framework. In this scheme, 
% one round is utilized to construct a noisy graph, after which vertices combine their local neighborhoods with this noisy graph. Subsequently, they report useful graph statistics to be aggregated by the central data curator. 
% In doing so, we can utilize the local neighbors of u and w and consider the two-hop paths with one edge in $G$.  
By adopting this multiple-round framework, we propose a \advsslong algorithm (\advss) where we only rely on the local view of \vq to estimate $\pqx$. 
Then, we introduce the \advdslong algorithm (\advds), which leverages the local neighborhoods of both \vq and \vx to optimize privacy budget allocation and balance the contribution of query vertices, resulting in minimized L2 loss.
% Then, we propose a \advdslong algorithm (\advds) that incorporates the local neighborhoods of both \vq and \vx, which optimizes the privacy budget allocation and balances the contribution of both query vertices for minimized L2 loss. 

% which further improves data utility by 
% (1) incorporating the local views of both \vq and \vx and 
% (2) adopting budget allocation optimization strategies to 
% distribute the privacy budget to noisy graph construction and Laplace mechanism for minimized L2 loss. 
% \textcolor{red}{mention something about the weighted average of the two.}
% constructing the noisy graph in the first round and the Laplace mechanism in the second round. 
% In addition, we discuss how our \advds algorithm can be further optimized for handling vertex pairs with imbalanced degrees. 

\begin{figure}[thb]
\centering  
\includegraphics[width=0.45\textwidth]{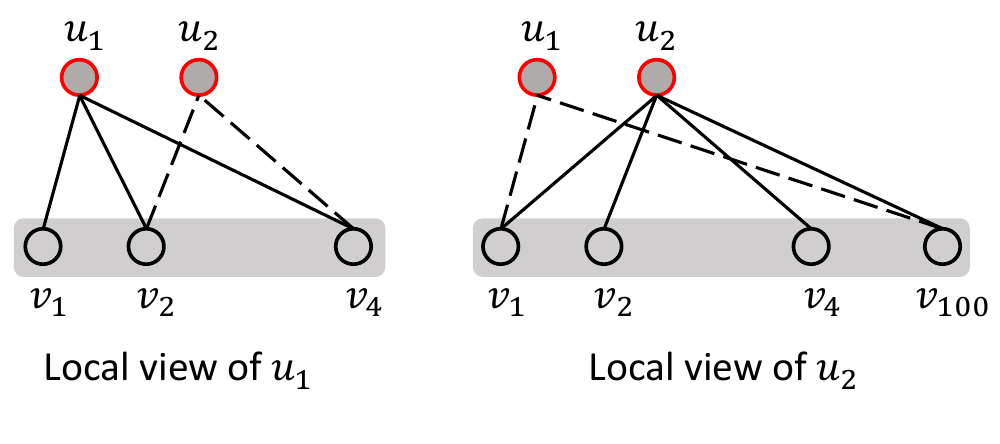}
\myspace
\caption{
% A bipartite graph and its corresponding noisy graph by applying randomized responses to $u_1$ and $u_2$. 
The construction of $\fq$ and $\fx$ based on the local neighborhoods of $u_1$ and $u_2$ ($u = u_1$, $w = u_2$).
% Two user-item networks and their corresponding noisy graphs constructed via randomized responses. 
}
\myspace
\myspace
% \myspace
\label{fig:adv}
\end{figure}
% \subsection{An unbiased estimate for $\pqx$} 

\subsection{A single-source estimator for $\pqx$} 
% Under the multiple-round framework, the privacy budget is divided into separate rounds. Within each round, the algorithm must adhere to $\varepsilon_k$-edge local differential privacy ($\varepsilon_k$-LDP). 
% Due to the sequential composition property of local differential privacy (LDP), the algorithm satisfies $\varepsilon$-LDP. 
In this part, we introduce a two-round algorithm for estimating $\pqx$. 
% where we allocate privacy budgets $\varepsilon_1$ to construct a noisy graph by applying randomized responses to \vq and \vx and $\varepsilon_2$ ($\varepsilon_1 + \varepsilon_2 = \varepsilon$). 
First, $\varepsilon_1$ is utilized to construct a noisy graph by applying randomized responses to \vq and \vx. 
Then, the data curator releases the noisy graph. In the second round, vertex \vq integrates its local neighbors with the noisy graph to derive an unbiased estimator for $\pqx$. Then, $\varepsilon_2$ is employed to apply the Laplace mechanism to add noise to this estimator. 

Now we assume that the noisy graph $G_{\varepsilon_1}'$ has already been constructed and discuss how to estimate $\pqx$ based on the local neighbors of \vq and the noisy neighbors of \vx. 
We start by noting that $\pqx$ can also be written as $\sum_{v \in N(u)} \mathcal{A}[v,w] $. 
Thus, when the neighbors of \vq are available, estimating $\pqx$ reduces to estimating $\mathcal{A}[v,w]$, which has already been addressed by $\phi(v, w) = \frac{\mathcal{A}'[v,w]-p}{1-2p}$ in Section \ref{sec:bs}. Here the flipping probability becomes $\frac{1}{1 + e^{\varepsilon_1}}$. 
Based on the above analysis, we derive the following unbiased estimate of $\pqx$ as 
\begin{align*}
   & f_u(u, w) = \sum_{ v \in N(u, G) } \phi(v, w)  = 
   \sum_{ v \in N(u, G) } \frac{\mathcal{A}'[v,w]-p}{1-2p} \\ 
   & = |N(u, G) \cap N(w, G_{\varepsilon_1}') | \frac{1-p}{1-2p} -|N(u, G) \setminus N(w,G_{\varepsilon_1}') |\frac{p}{1-2p} 
\end{align*}

At this point, $f_u(u, w)$ is computed locally based on the neighborhood of \vq. To release it under edge LDP, we analyze the global sensitivity of $f_u(u, w)$ and apply the Laplace mechanism. 

\noindent
{\bf Global Sensitivity Analysis.}
By Definition \ref{def:gs}, the global sensitivity of $f_u(u, w)$ is defined as $\Delta(f_u(u, w)) = \max_{u, u' \in V(G)} | f_u(u, w) - f_{u'}(u, w) |$, 
where $u'$ represents a hypothetical vertex differing from \vq in its neighbor list at one entry. 
It follows:
\begin{align*}
\Delta(f_u(u, w)) 
% & \leq \left| \sum_{v \in N(u, G)} \phi(v, w) - \sum_{v \in N(u', G)} \phi(v, w) \right| \\
\leq \max_{v'} | \phi(v', w) | = \frac{1-p}{1-2p}
\end{align*}

The last step is because the absolute value of $\phi(v', w)$ is either $\frac{1-p}{1-2p}$ or $\frac{p}{1-2p}$. Since $p = \frac{1}{1+ e ^{\varepsilon}} < \frac{1}{2}$, $\frac{1-p}{1-2p}$ is always larger than $\frac{p}{1-2p}$. 
This bound suggests that we must add Laplacian noise scaled to $\frac{1-p}{1-2p}$ before sending the estimator to the data curator. 
In other words, the data curator receives the noisy version of $f_u(u, w)$ denoted by 
\begin{align}
\fq(u,w) = \sum_{v \in N(u, G)} \frac{\mathcal{A}'[v,w]-p}{1-2p} + \text{Lap}\left(\frac{1-p}{(1-2p) \varepsilon_2}\right)
\label{eq:fq}
\end{align}

\begin{lemma}
\label{thm:fq.biase}
$\fq(u,w)$ in Equation \ref{eq:fq} is an unbiased estimate of $\pqx$, i.e., 
$\mathbb{E}(\fq(u,w)) = \pqx$. 
\end{lemma}
\begin{proof}
% Since $\phi(v,w) = \frac{\mathcal{A}'[v,w]-p}{1-2p}$ is an unbiased estimator for $\mathcal{A}[v,w]$, 
Since $\mathbb{E}(\frac{\mathcal{A}'[v,w]-p}{1-2p})  = \mathcal{A}[v,w]$, the expected value of the first term in $\fq$ is $\mathbb{E}( \sum_{v \in N(u, G)} \mathcal{A}[v,w]) = \pqx$. 
The second term represents the noise drawn from the Laplacian distribution with an expected value of zero. Thus, $\fq(u,w)$ is unbiased. 
\end{proof}

\begin{algorithm}[tbh]
    \small
	\caption{The \advss algorithm}
	\label{algo:adv1}
	\LinesNumbered
	\KwIn{$G$: a bipartite graph; 
            $\varepsilon$: a privacy budget; 
            \vq, \vx: two query vertices} 
         \KwOut{$\fq(u, w) $}
        split privacy budget $\varepsilon$ into $\varepsilon_1$ and $\varepsilon_2$ evenly;\\
        \tcp{\textbf{round 1:}}
        $p \gets \frac{1}{1+ e^{\varepsilon_1}}$;\\
\ForEach{$j \in$ the opposite layer from $u$ and $w$}{
perturb 
    $\mathcal{A}'[u,j]  \gets
    \begin{cases}
        1- \mathcal{A}[u,j] , &  \text{w.p. }  \frac{1}{1+ e^{\varepsilon}} \\
        \mathcal{A}[u,j] & \text{w.p. }  \frac{e^{\varepsilon}}{1+ e^{\varepsilon}}
    \end{cases}
    $
} 
send noisy edges to the data curator;\\
$G_{\varepsilon_1}' \gets $ the noisy graph constructed from $\mathcal{A}'[u,j]$;\\
     \tcp{\textbf{round 2:}}
            $S_1 \gets 0$; $S_2 \gets 0$;\\
            \ForEach{$v \in N(u, G)$ }{
                \If{ $(v, w) \in E(G_{\varepsilon_1}')$  }{
                    $S_1 \gets S_1 + 1$:\\
                }
                \Else{
                    $S_2 \gets S_2 + 1$:\\
                }
            }
        $\fq(u, w) \gets S_1 \times \frac{1-p}{1-2p} - S_2 \times \frac{p}{1-2p}$;\\
	% \textbf{return} $\fq(u, w)$;\\
        $ \fq(u, w) \gets \fq(u, w) + \text{Lap}\left(\frac{1-p}{(1-2p) \varepsilon_2}\right)$;\\
	\textbf{return} $\fq(u, w) $;\\
\end{algorithm}

\noindent
{\bf The \advss algorithm.} 
In this part, we present the \advsslong algorithm (\advss) which involves two rounds of interaction between the vertices and the data curator and returns the unbiased estimator $\fq(u ,w)$ derived in Lemma \ref{thm:fq.biase}. The detailed steps are summarized in Algorithm \ref{algo:adv1}. 
Initially, \advss splits the privacy budget $\varepsilon$ into $\varepsilon_1$ and $\varepsilon_2$ evenly. 
In the first round, randomized responses are applied to both \vq and \vx to generate noisy edges, which are then transmitted to the data curator (Lines 3-6). Then, $G_{\varepsilon_1}'$ is constructed from these noisy edges. 
In the second round, \advss visits the neighbors of \vq on $G$ and counts how many are connected to \vx on $G_{\varepsilon_1}'$. 
Upon termination of the for-loop (Lines 9-13), $S_1$ represents the $|N(u, G) \cap N(w, G_{\varepsilon_1}') |$ and $S_2$ represents the $|N(u, G) \setminus N(w,G_{\varepsilon_1}') |$. 
Based on $S_1$ and $S_2$, \advss computes $\fq(u, w)$ and add Laplacian noise scaled to $\frac{1-p}{(1-2p) \varepsilon_2}$ (Lines 14, 15). 
Compared to \bs that considers all vertices on the opposite layer from \vq and \vx on the noisy graph, \advss limits the search scope for the common neighbors of \vq and \vx to the local neighbors of \vq, which results in a substantial reduction in L2 loss. 

\begin{example}
Consider the bipartite graph in Fig.~\ref{fig:motivation}. 
In Fig.\ref{fig:adv}, we illustrate the construction of the single-source estimators where $u = u_1$ and $w = u_2$. 
The outlines of $u_1$ and $w_2$ are highlighted in red. 
A privacy budget $\varepsilon_1$ is allocated to randomized responses for $u_1$ and $u2$. 
In Fig.~\ref{fig:adv}, each query vertex can download the noisy edges from the other query vertex and integrate them with its neighbors. 
In the local perspective of $u_1$, the solid lines represent edges between $u_1$ and its neighbors, while the dashed lines represent noisy edges from $u_2$. 
The vertices shaded in grey are candidates for common neighbors between $u_1$ and $u_2$ in \advss, which includes the neighbors of $u_1$ in the original graph. 
Note that this is much smaller than the candidate pool identified by \bs, which includes all vertices on the opposite layer from the query vertices. 
Based on the formula for $\fq$, it only relies on three Bernoulli variables: $\mathcal{A}'[v_1, u_2]$, $\mathcal{A}'[v_2, u_2]$, and $\mathcal{A}'[v_4, u_2]$. 
Similarly, we can derive $\fx$ based on four random variables: $\mathcal{A}'[v_1, u_2]$, $\mathcal{A}'[v_2, u_2]$, $\mathcal{A}'[v_4, u_2]$, and $\mathcal{A}'[v_{100}, u_2]$. 
The reliance on fewer random variables accounts for the smaller expected L2 loss of \advss compared to \bs.
\end{example}

\noindent
{\bf Theoretical analysis for \advss.} 
{
Without loss of generality, we assume that the query vertices $u$ and $w$ are in $L(G)$. 
First, we analyze the computational time complexity of Algorithm \ref{algo:adv1} (\advss). On the vertex side, the time costs incurred by randomized responses is $O(n_1)$, where $n_1 = |U(G)|$. On the curator side, visiting the neighbors of $u$ in $G$ to compute $\fq(u, w)$ takes $O(deg(u, G) $ time. The overall time complexity is $O(n_1)$. 
% The communication costs of \bs are incurred only during randomized responses, similar to \naive. The overall communication cost is $O\left(\frac{e^{\varepsilon}-1}{e^{\varepsilon}+1}(d_u + d_w) + \frac{2n_1}{1 + e^{\varepsilon}}\right)$. 
Then, we analyze the communications costs of \advss, which include (1) sending the noisy edges from $w$ and downloading them to vertex $u$ and (2) sending the single-source estimator $\fq(u, w)$ to the data curator. 
The dominating cost is incurred by the step (1). 
Note that the expected number of noisy edges from vertex $w$ is $d_w \times (1-p) + (n_1 - d_w) \times p$, where $p = \frac{1}{1 + e^{\varepsilon_1}}$. 
Thus, the overall communication cost is $O(\frac{e^{\varepsilon_1}-1}{e^{\varepsilon_1}+1}d_w + \frac{ n_1}{1 + e^{\varepsilon_1}} )$. 
In the following, we analyze the expected L2 loss of \advss. }

% We analyze the L2 loss of Algorithm \ref{algo:adv1} as follows. 
\begin{theorem}
\label{thm:fq}
The expected L2 loss of $\fq(u,w)$ in Equation \ref{eq:fq} is 
% $\frac{p(1-p)}{(1-2p)^2} d_u  + \frac{2(1-p)^2}{(1-2p)^2\varepsilon_2^2}$
$  O(\frac{ e^{\varepsilon_1}}{(1 - e^{\varepsilon_1})^2} (d_u + \frac{ 2 e^{\varepsilon_1}}{\varepsilon_2^2} ) )$. Here $d_u$ represents the degree of \vq in $G$. 
% $p = \frac{1}{1 + e^{\varepsilon_1}}$ represents the flipping probability from randomized responses. 
\end{theorem}
\begin{proof}
Since $\fq(u,w)$ is an unbiased estimator, its L2 loss equals its variance. 
The variance of $\fq(u,w)$ consists of two parts: $f_u(u,w)$ and the Laplacian noise. 
First, it immediately follows that the variance from $f_u(u,w)$ is $\frac{p(1-p)}{(1-2p)^2} d_u $ based on Equation \ref{eq:phi.var}. This is because each $\phi(u, w)$ is independent of each other and there is no covariance involved. 
In addition, the variance from the Laplacian noise is: 
$$
\mathrm{Var}(Lap(\frac{\Delta(\fq(u, w))}{\varepsilon_2})) = 2 (\frac{1-p}{(1-2p)\varepsilon_2})^2
 = \frac{2(1-p)^2}{(1-2p)^2\varepsilon_2^2}
$$
% Putting these together completes the proof. 
Since $f_u(u,w)$ and the Laplacian noise are independent, the expected L2 loss of $\fq$ is
$\frac{p(1-p)}{(1-2p)^2} d_u  + \frac{2(1-p)^2}{(1-2p)^2\varepsilon_2^2}$. 
Substituting $p = \frac{1}{1+ e^{\varepsilon_1}}$ into the above expression completes the proof. 
\end{proof}
{Since the L2 loss of $\fq(u,w)$ is no longer dependent on $n_1$, the data utility of \advss is significantly improved compared to \bs.} 
% local differential privacy inherits the sequential and parallel composition features mentioned in Section 2.4.
% In addition, to verify that the \advss algorithm satisfies edge LDP, we introduce the following {\em sequential composition} theorem. 
In addition, we check whether Algorithm \ref{algo:adv1} satisfies the privacy requirements of \epldp in the following theorem. 
%%% parallel composition theorem
% \begin{theorem}
% {\bf Sequential Composition \cite{jiang2021applications}.}
% \label{thm:sequentialcomposition}
% Consider $t$ edge LDP algorithms $Z_1, Z_2, \ldots, Z_t$, whose privacy budgets are denoted by $\varepsilon_1, \varepsilon_1, \ldots, \varepsilon_t$, respectively. We have the following property. 
% The composite algorithm obtained by sequentially applying $Z_1, Z_2, \ldots, Z_t$ on the bipartite graph $G$ provides $\sum_{i=1}^{t} \varepsilon_i$-differential privacy.
% \end{theorem}
%  which states that 
% when we sequentially apply various differential privacy algorithms to a 
% dataset
\begin{theorem} 
{Given a bipartite graph $G$ and a privacy budget $\varepsilon$, Algorithm \ref{algo:adv1} satisfies \epldp. }
% \textcolor{red}{TODO: provide formal proof. }
\end{theorem} 
\begin{proof}
{
We use the {\em Sequential Composition} theorem \cite{jiang2021applications} to prove that Algorithm \ref{algo:adv1} satisfies \epldp. 
% Since randomized responses provide \epldp, the first round of Algorithm \ref{algo:adv1} satisfies $\varepsilon_1$-edge LDP (Lines 2-5). 
In the first round, generating the noisy edges via randomized responses satisfies $\varepsilon_1$-edge LDP (Lines 2-6). 
In the second round, Lines 7-13 are conducted locally by the vertex $u$. 
Then, the Laplace mechanism (Line 14) is applied w.r.t. a privacy budget of $\varepsilon_2$ to construct the unbiased estimator $\fq$, which satisfies $\varepsilon_2$-edge LDP. 
% the query vertex $u$ constructs an unbiased estimator $\fq$ by applying the Laplace mechanism w.r.t. a privacy budget of $\varepsilon_2$, which satisfies $\varepsilon_2$-edge LDP. 
By the sequential composition property of edge LDP, Algorithm \ref{algo:adv1} satisfies \epldp ($\varepsilon = \varepsilon_1 + \varepsilon_2$). }
\end{proof}

\subsection{A double-source estimator for $\pqx$} 
\begin{figure}[tb]
\centering  
\includegraphics[width=0.23\textwidth]{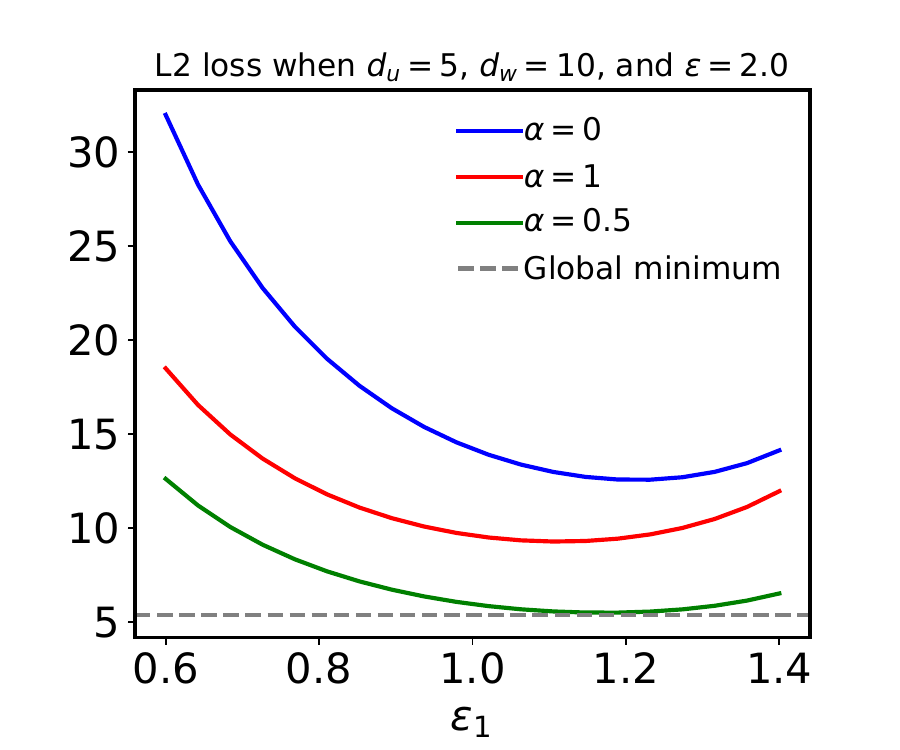}
\includegraphics[width=0.23\textwidth]{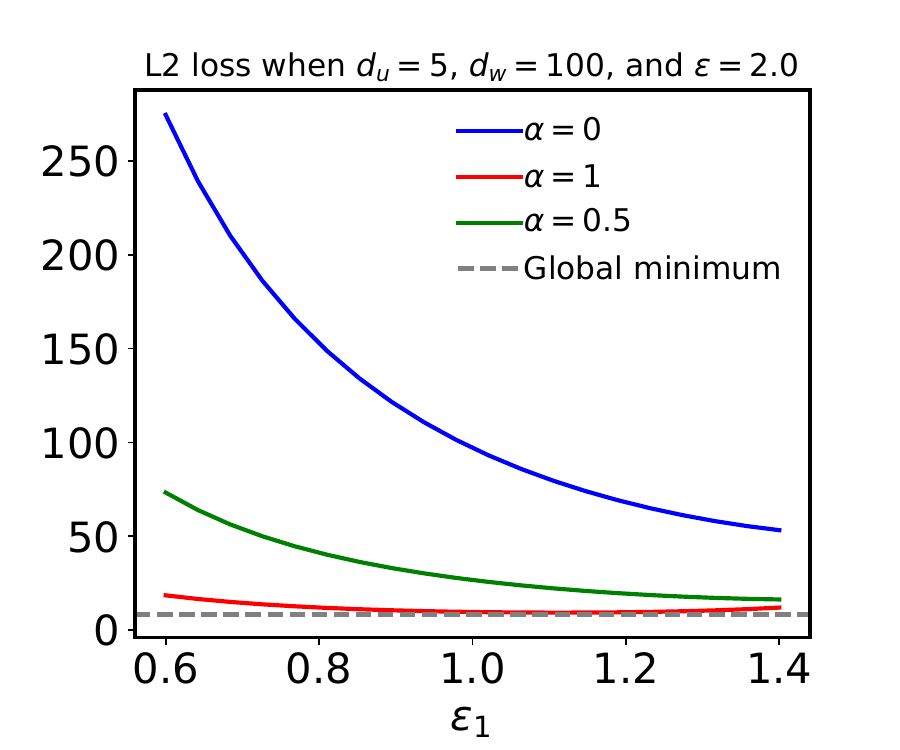}
% \includegraphics[width=0.4\textwidth]{figures/fig2.pdf}
% \myspace
\myspace
\myspace
\caption{Illustration of the L2 loss of $\fds$ when $\varepsilon = 2$.}
\label{fig:variancecompare}
\end{figure}

% \textcolor{red}{TODO: add some illustrations}

The single-source estimator $\fq(u,w)$ only involves the neighborhood of \vq. Similarly, we can develop another unbiased estimator, $\fx(u,w)= \sum_{ v \in N(w, G) } \phi(v, u)$, by applying the same process to the neighborhood of \vx. 
% \textcolor{red}{give the formula for $\fx$}
This raises a natural question: How can we integrate these estimators to further minimize L2 loss while maintaining unbiasedness? 
Examining the loss of L2 of $\fq(u,w)$ in Theorem \ref{thm:fq}, we can see that it consists of the first term representing the error incurred by randomized responses, and the second term representing the error incurred by Laplacian noise. 
On the one hand, if we only minimize the first term, we could always choose the estimator between $\fq$ and $\fx$ whose corresponding query vertex has a smaller degree. 
On the other hand, if we only focus on minimizing the second term, we could take an average of $\fq$ and $\fx$ and the Laplacian noise of the resulting estimator will be reduced by half. 
To balance both objectives, we propose a double-source estimator $f^*$ by taking a weighted average of $\fq$ and $\fx$, i.e., $f^* = \alpha \fq+ (1 - \alpha) \fx$ ($\alpha \in [0,1]$). 
Here $\alpha$ is the weighting parameter that adjusts the contribution of $\fq$ and $\fx$. 
By analyzing the L2 loss of $f^*$, we introduce the \advdslong algorithm (\advds), which enhances data utility by optimizing the allocation of privacy budget and balancing the contribution of $\fq$ and $\fx$. 

\noindent
{\bf Properties of the double-source estimator $\fds$.} 
% \noindent
% {\bf  Unbiasedness of the double-source estimator $\fds$.} 
Given that $\fds$ is a weighted average of $\fq$ and $\fx$, its unbiasedness directly stems from the principle of linearity in expected values, i.e., $\mathbb{E}(aX + bY) = a\mathbb{E}(X) + b\mathbb{E}(Y)$. 
Thus, based on the bias-variance decomposition, its L2 loss equals its variance, as analyzed in the following theorem. 
\begin{theorem}
\label{thm:balance}
The L2 loss of $f^* = \alpha \fq+ (1 - \alpha) \fx$ ($\alpha \in [0,1]$) is $\frac{ e^{\varepsilon_1}}{(1 - e^{\varepsilon_1})^2} \left( (\alpha^2 d_u+(1-\alpha)^2d_w)  + \frac{ 2(\alpha^2 + (1-\alpha)^2)e^{\varepsilon_1}}{\varepsilon_2^2} \right)
$.
% $\frac{p(1-p)}{(1-2p)^2} (\alpha^2 d_u + (1 - \alpha)^2 d_w) +
% \frac{2(1-p)^2}{(1-2p)^2\varepsilon_2^2}(\alpha^2 + (1 - \alpha)^2)$. 
Here $d_u$ and $d_w$ represent the degrees of \vq and \vx in $G$. 
% $p = \frac{1}{1 + e^{\varepsilon_1}}$ represents the flipping probability from randomized responses. 
\end{theorem}
\begin{proof}

\begin{align*}
    & l2(f^*, \pqx)  = \text{Var}(f^*) = \alpha^2 \text{Var}(\fq) + (1 - \alpha)^2 \text{Var}(\fx) \notag \\
    % & = \alpha^2\left( \frac{p(1-p)}{(1-2p)^2} d_u + \frac{2(1-p)^2}{(1-2p)^2\varepsilon_2^2} \right) + \\
    % & (1 - \alpha)^2 \left( \frac{p(1-p)}{(1-2p)^2} d_w+ \frac{2(1-p)^2}{(1-2p)^2\varepsilon_2^2} \right) \\
    & = \frac{p(1-p)}{(1-2p)^2} (\alpha^2 d_u + (1 - \alpha)^2 d_w) +
     \frac{2(1-p)^2}{(1-2p)^2\varepsilon_2^2}(\alpha^2 + (1 - \alpha)^2)
     \label{xxx}
\end{align*}  
{
Note that $\tilde{f_u}$ depends on the noisy edges connected to $w$, while $\tilde{f_w}$ depends on the noisy edges connected to $u$. Since they depend on disjoint edges in the noisy graph, $\tilde{f_u}$  and $\tilde{f_w}$ are independent, and their covariance $Cov(\fq, \fx)=0$. Thus, the first step holds. }
% Thus, $\tilde{f_u}$  and $\tilde{f_w}$  are independent and the first step holds. }
% (2) the noisy edges connected to u. 
% The first step holds because $\fq$ and $\fx$ are independent and their covariance $Cov(\fq, \fx)=0$. 
Also, the expected L2 loss of $\fx$ is the same as $\fq$ where $d_u$ is replaced by $d_w$. 
Substituting $p = \frac{1}{1+ e^{\varepsilon_1}}$ completes the proof. 
\end{proof}

{Based on Theorem \ref{thm:balance}, the L2 loss of $\fds$ is a function of $\varepsilon_1$ and $\alpha$, where $d_u$ and $d_w$ are constants. 
We denote it by $l2(f^*, \pqx) := F(\varepsilon_1, \alpha)$. 
To obtain estimates for $d_u$ and $d_w$, we can apply the Laplace mechanism in an additional round using a small privacy budget ($\varepsilon_0$). 
If the estimates of $d_u$ and $d_w$ are negative, we can estimate the average vertex degree in $L(G)$ and substitute them.} 
To minimize this loss, we seek values of $\varepsilon_1 \in (0, \varepsilon)$ and $\alpha \in [0,1]$ that minimize $F(\varepsilon_1, \alpha)$. 
We discover that $F$ reaches its global minimum if and only if its partial derivatives $\frac{\partial F}{\partial \alpha} = \frac{\partial F}{\partial \varepsilon_1} =0$. 
However, this results in a transcendental equation that lacks analytical solutions. 
Thus, we resort to Newton's method \cite{galantai2000theory} for high-precision approximate solutions. 
By optimizing $\varepsilon_1$ and $\alpha$, the resulting L2 loss of $\fds$ will be lower than that of both single-source estimators $\fq$ and $\fx$. 

{
We could also optimize the expected L2 loss for \advss, which is a function of $\epsilon_1$ and $deg(u)$. 
Specifically, we could spend a small privacy budget ($\varepsilon_0$) to estimate $deg(u)$ and then apply Newton's method to find the best privacy budgets ($\varepsilon_1$ and $\varepsilon_2$) that minimize the expected L2 loss of \advss ($\varepsilon = \sum_{i=0}^2 \varepsilon_i $). 
In practice, this implementation only outperforms the current \advss with an even separation of privacy budget ($\varepsilon_1 = \varepsilon_2$) when the degree of $u$ is large. 
In addition, it is a special case of \advds where $\alpha =1$. 
}

\begin{algorithm}[tbh]
    \small
	\caption{The \advds algorithm}
	\label{algo:adv2}
	\LinesNumbered
	\KwIn{
            $G$: a bipartite graph; 
            $\varepsilon$: a privacy budget; 
            \vq, \vx: two query vertices
        }
         \KwOut{$\fq(u, w) $}
        % split privacy budget $\varepsilon$ into  $\varepsilon_1$ and $\varepsilon_2$;\\
        \tcp{\textbf{round 1:}}
        $\varepsilon_0 \gets \varepsilon \times 0.05$;\\
        % $\varepsilon_0 \gets$ a small part of $\ep$;\\
        $ d_u \gets deg(u, G) + \text{Lap}(\frac{1}{\varepsilon_0})$;\\
        $ d_w \gets deg(w, G) + \text{Lap}(\frac{1}{\varepsilon_0})$; \\
        $d' \gets$ the average vertex degree on the same layer as \vq;\\
        correct $d_u$ and $d_w$ with $d'$;\\
        find $\alpha$ and $\varepsilon_1$ that minimizes $\text{Var}(f^*)$;\\
        \tcp{\textbf{round 2:}}
        $p \gets \frac{1}{1+ e^{\varepsilon_1}}$;\\
\ForEach{$i \in \{u, w\}$}{
    \ForEach{$j \in$ the opposite layer from $u$ and $w$}{
    perturb 
        $\mathcal{A}'[i,j]  \gets
        \begin{cases}
            1- \mathcal{A}[i,j] , &  \text{w.p. }  \frac{1}{1+ e^{\varepsilon}} \\
            \mathcal{A}[i,j] & \text{w.p. }  \frac{e^{\varepsilon}}{1+ e^{\varepsilon}}
        \end{cases}
        $
    } 
}
send noisy edges to the data curator;\\
$G_{\varepsilon_1}' \gets $ the noisy graph constructed from $\mathcal{A}'[i,j]$;\\
         \tcp{\textbf{round 3:}}
            $\varepsilon_2 \gets \varepsilon - \varepsilon_0 -\varepsilon_1$;\\
            $\fq(u, w) \gets $ the estimator by running Lines 8-15 of Algorithm \ref{algo:adv1};\\
            $\fx(u, w) \gets $ the estimator by running Lines 8-15 of Algorithm \ref{algo:adv1} with \vq and \vx switched;\\                     
	\textbf{return} $\alpha \fq + (1 - \alpha) \fx$;\\
\end{algorithm}

\noindent
{\bf The \advds algorithm. }
In this part, we present the \advdslong algorithm (\advds) which uses an additional round compared to \advss to estimate $d_u$ and $d_w$ and estimate the L2 loss of $\fds$. 
The detailed steps are outlined in Algorithm \ref{algo:adv2}. 
% The \advss algorithm is designed to estimate $P_2(u, w)$ for all vertices \vx on the same layer as the query vertex \vq. It takes as input a bipartite graph $G$, a privacy budget $\varepsilon$, and the query vertex \vq. The output is $P_2(u, w)$ for all vertices \vx on the same layer as \vq.
In the first round, \advds uses a small privacy budget $\varepsilon_0$ and applies the Laplace mechanism to obtain unbiased estimates of $d_u$ and $d_w$ (Lines 1-3). 
{Here the global sensitivity of $d_u$ ($d_w$) is one because adding or deleting an edge from the neighbor list of \vq (\vx) changes $d_u$ ($d_w$) by at most one. }
Due to the Laplacian noise, the reported $d_u$ and $d_w$ could be negative. In this case, we correct for any negative value with the estimated average degree of the vertices on the same side as \vq and \vx (Lines 4, 5). 
Then, the \advds algorithm invokes Newton's method to find the pair of $\alpha$ and $\varepsilon_1$ that minimizes the estimated L2 loss of $\fds$. 
In the second round, randomized responses are applied to \vq and \vx with respect to $\varepsilon_1$, leading to the noisy graph $G_{\varepsilon_1}'$ (Lines 7-12). 
% In the third round, \advds uses the remaining privacy budget $\varepsilon_2$ to build the unbiased estimators $\fq$ and $\fx$ from the local neighborhoods of \vq and \vx. 
% Specifically, $\fq$ is obtained by executing Lines 8-15 of \advss. 
% Similarly, $\fx$ is computed by executing Lines 8-15 of \advss by visiting the neighbors of \vx in $G$. 
In the third round, \advds allocates the remaining privacy budget $\varepsilon_2$ to construct unbiased estimators $\fq$ and $\fx$ from the local neighborhoods of \vq and \vx. Specifically, $\fq$ is derived by executing Lines 8-15 of \advss, while $\fx$ is computed similarly by visiting the neighbors of \vx in $G$ (Lines 14-15). 
Note that when constructing $\fq$ and $\fx$, the global sensitivity analysis is the same as in \advss. In other words, the global sensitivity of $\frac{1-p}{1-2p}$ for each single-source estimator is applied to both $\fq$ and $\fx$ upon construction. 
In the end, \advds returns the weighted average of $\fq$ and $\fx$ where the parameter $\alpha$ is computed in the first round (Line 16). 
% Note that Algorithm \ref{algo:adv2} satisfies \epldp by the sequential composition property in Theorem \ref{thm:sequentialcomposition} ($\varepsilon = \varepsilon_0 +\varepsilon_1 + \varepsilon_2$). 

{
\noindent
{\bf Theoretical analysis for \advds.} 
Without loss of generality, we assume that $u$ and $w \in L(G)$. 
First, we analyze the computational time complexity of \advds. 
On the vertex side, estimating the average degree of the vertices in $L(G)$ takes $O(n_2)$ time. 
When constructing the noisy graph, the time costs incurred by the randomized responses are $O(n_1)$. 
On the curator side, visiting the neighbors of $u$ and $w$ to compute $\fq$ and $\fx$ takes $O(deg(u, G)  + deg(w, G) $ time. 
Thus, the overall time complexity is $O(n)$.}

{We then analyze the communication costs of \advds, which include: 
    (1) sending the noisy degree of all vertices in $L(G)$,
    (2) sending the noisy edges from $w$ and downloading them to vertex $u$,
    (3) sending the noisy edges from $u$ and downloading them to vertex $w$, and
    (4) sending two single-source estimators $\fq$ and $\fx$ to the data curator. 
Step (1) incurs communication costs of $O(n_2)$. 
The communication costs for Step (2) and Step (3) are proportional to the expected number of noisy edges from $u$ and $w$, which is $(d_u + d_w) \times (1-p) + 2(n_1 - d_w) \times p$, where $p = \frac{1}{1 + e^{\varepsilon_1}}$. 
Step (4) incurs a communication cost of $O(1)$. 
Thus, the overall communication cost is $O(n_2 + \frac{e^{\varepsilon_1}-1}{e^{\varepsilon_1}+1} (d_w+d_u) + \frac{ 2 n_1}{1 + e^{\varepsilon_1}})$. 
}

{
Since the expected L2 loss of \advds has been analyzed in Theorem \ref{thm:balance}, we compare it with the expected L2 loss of \advss in the following theorem. 
\begin{theorem}
\label{thm:compare}
The minimum L2 loss incurred by the double-source estimator $\fds= \alpha \fq + (1 - \alpha) \fx$ is less than or equal to the L2 loss incurred by both single-source estimators $\fq$ and $\fx$. 
$$
\min_{\varepsilon_1, \alpha} \ l2(\fds, \pqx) \leq \min(l2(\fq, \pqx), l2(\fx, \pqx))
$$
\end{theorem}
\begin{proof}
    Let $L^*$ be the minimized expected L2 loss of the double-source estimator $\fds$. 
    To prove the above inequality, we need to prove that $L^* \leq l2(\fq, \pqx)$ and $L^* \leq l2(\fx, \pqx)$. 
    By construction, $\fq$ is a special case of $\fds$ where $\alpha =1$, i.e., $\fq = \fds |_{\alpha =1}$. 
    Hence, for any privacy budget allocations ($\varepsilon_1$ and $\varepsilon_2$), we have $L^* \leq l2(\fq, \pqx)$. 
    Similarly, $\fx$ is also a special case of $\fds$ where $\alpha =0$, i.e., $\fq = \fds |_{\alpha =0}$. 
    We also obtain $L^* \leq l2(\fx, \pqx)$. 
    Combining these two inequalities completes the proof. 
\end{proof}
}

To illustrate the comparison, we plot the L2 loss of $\fds$, $\fq$, and $\fx$ against varying $\varepsilon_1$ values in Fig.~\ref{fig:variancecompare}, where $\varepsilon=2$. 
The blue curve labeled ``$\alpha = 0$'' represents the L2 loss of $\fx$. 
The red curve labeled ``$\alpha = 1$'' represents the L2 loss of $\fq$. 
The green curve labeled ``$\alpha = 0.5$''  represents the unbiased estimator $f' = (\fq + \fx)/2$. 
The grey horizontal line marks the global minimum L2 loss for $\fds$. 
On the left, when $d_u = 5$ and $d_w = 10$, $f'$ outperforms $\fq$ and $\fx$ and reaches the global minimum. 
On the right, when $d_u$ and $d_w$ are more imbalanced, $\fq$ becomes the best estimator, reaching the global minimum. 
None of the single-source estimators or their average can consistently reach the minimized L2 loss of $\fds$ for all query vertex pairs. 
This is due to the flexibility of $\fds$ in adjusting the privacy budget allocation and the weighting of $\fq$ and $\fx$.

In the following theorem, we verify the compliance of Algorithm \ref{algo:adv2} to \epldp. 
\begin{theorem} 
{
Given a bipartite graph $G$ and a privacy budget $\varepsilon$, Algorithm \ref{algo:adv2} satisfies \epldp. }
% \textcolor{red}{TODO: provide formal proof. }
\end{theorem} 
\begin{proof}
{
We use the {\em Sequential Composition} and {\em Parallel Composition} theorems for \epldp \cite{jiang2021applications}. 
Parallel composition theorem states that if different \epldp algorithms are applied to disjoint datasets with privacy budgets \(\varepsilon_i\), the composite algorithm satisfies \(\max_i \varepsilon_i\)-LDP. 
In the first round, each vertex reports its degree using the Laplace mechanism, achieving \(\varepsilon_0\)-edge LDP. 
By parallel composition, this round satisfies \(\varepsilon_0\)-edge LDP. 
In the second round, randomized responses provide \(\varepsilon_1\)-edge LDP. 
In the third round, building \(\fq\) and \(\fx\) satisfies \(\varepsilon_2\)-edge LDP. 
By parallel composition, this round satisfies \(\varepsilon_2\)-edge LDP. 
By sequential composition, Algorithm \ref{algo:adv2} satisfies \(\epsilon\)-LDP with \(\varepsilon = \varepsilon_0 + \varepsilon_1 + \varepsilon_2\). 
% Thus, the theorem holds.
}
\end{proof}

\noindent
{\bf Summary of the expected L2 losses of all algorithms.} 
{
% In Table \ref{tab:complexitycompare}, we summarize the computational time complexities, expected L2 losses, and communication costs of all privacy-preserving algorithms for the estimation of the common neighborhood. 
In Table \ref{tab:complexitycompare}, we summarize the expected L2 losses of all privacy-preserving algorithms for estimating the number of common neighbors. 
%%% compare the L2 loss of all algorithms 
% The expected L2 loss of \bs is smaller than \naive, being $O(n_1)$ instead of $O(n_1^2)$. 
The expected L2 loss of \bs is smaller than that of \naive, with \bs having an expected L2 loss of $O(n_1)$ compared to \naive's $O(n_1^2)$. 
In addition, the expected L2 losses of \advss and \advds are lower than those of \naive and \bs because they do not depend on the number of vertices in the graph. 
Between \advds and \advss, as analyzed in Theorem 9, the minimized loss for \advds is smaller than \advss because \advss is a special case of \advds where $\alpha = 0$ or $\alpha = 1$. 
% We also summarize the computational time complexities and communication costs of these algorithms. We will evaluate their performance through experiments.
Note that for \bs, \advss, and \advds, which are unbiased estimators, their expected L2 losses can offer insight into their deviation from the true value by applying \emph{Chebyshev's inequality} \cite{saw1984chebyshev}. 
For instance, for the \bs algorithm, we know $\mathbb{E}(\foneround(u, w)) = \pqx$ and that $\mathrm{Var}(\foneround(u, w)) = \frac{p^2(1-p)^2}{(1-2p)^4} n_1
    + \frac{p(1-p)}{(1-2p)^2}(d_w + d_w ) $. 
Chebyshev's inequality states that for any $k >0$: 
$$
P\left(|\foneround(u, w) - \pqx| \geq k \sqrt{\mathrm{Var}(\foneround(u, w))} \right) \leq \frac{1}{k^2}.
$$
Similar probabilistic bounds can be derived for the \advss and \advds algorithms based on their expected L2 losses.}
% Let $X$ be a random variable with finite expected value $\mu = \mathbb{E}[X]$ and finite variance $\sigma^2 = \mathbb{V}[X]$. For any $k > 0$,

% \begin{equation}
% \Pr(|X - \mu| \geq k\sigma) \leq \frac{1}{k^2}.
% \end{equation}

% \input{tabcomplexity}
\begin{table}[tbh]
\centering
\caption{Summary of Datasets}
\scalebox{0.8}{
\begin{tabular}{c|cc|ccc}
\noalign{\hrule height 1.20pt}
\cellcolor{gray!25} Dataset & \cellcolor{gray!25}Upper & \cellcolor{gray!25} Lower & \cellcolor{gray!25} $|E|$ & \cellcolor{gray!25} $|U|$ & \cellcolor{gray!25} $|L|$ \\ 
\noalign{\hrule height 0.7pt}
Rmwiki (RM) & User & Article & 58.0K & 1.2K & 8.1K \\
Collaboration (AC) & Author & Paper & 58.6K & 16.7K & 22.0K \\
Occupation (OC) & Person & Occupation & 250.9K & 127.6K & 101.7K \\
% Location (LO) & Entity & Place & 293.7K & 172.1K & 53.4K \\
Bag-kos (DA) & Document & Word & 353.2K & 3.4K & 6.9K \\
Bpywiki (BP) & User & Article & 399.7K & 1.3K & 57.9K \\
Tewiktionary (MT) & User & Article & 529.6K & 495 & 121.5K \\
Bookcrossing (BX) & User & Book & 1.1M & 105.3K & 340.5K \\
Stackoverflow (SO) & User & Post & 1.3M & 545.2K & 96.7K \\
Team (TM) & Athlete & Team & 1.4M & 901.2K & 34.5K \\
% Digg-votes (DV) & User & Story & 3.0M & 139.4K & 3.6K \\
Wiki-en-cat (WC) & Article & Category & 3.8M & 1.9M & 182.9K \\
Movielens (ML) & User & Movie & 10.0M & 69.9K & 10.7K \\
Epinions (ER) & User & Product & 13.7M & 120.5K & 755.8K \\
Netflix (NX) & User & Movie & 100.5M & 480.2K & 17.8K \\
Delicious-ui (DUI) & User & Url & 101.8M & 833.1K & 33.8M \\
Orkut (OG) & User & Group & 327.0M & 2.8M & 8.7M \\
\noalign{\hrule height 1.20pt}
\end{tabular}}

\label{tab:summary}
\end{table}
\begin{table*}[tbh]
\centering
\caption{Summary of time costs, expected L2 losses, and communication costs}
        \myspace
        \myspace
\scalebox{0.9}{
\begin{tabular}{c|c|c|c|c} 
\noalign{\hrule height 1pt}
\cellcolor{gray!25} Algorithm   & \cellcolor{gray!25} Time cost &  \cellcolor{gray!25} Unbiased & \cellcolor{gray!25} Expected L2 loss & \cellcolor{gray!25} Communication cost\\
\noalign{\hrule height 0.56pt}
\naive    & $O(n_1)$ & $\times$  & \makecell{ \vspace{0.3mm} $\errnaive$ \vspace{0.3mm}}  & $O(\frac{e^{\varepsilon}-1}{e^{\varepsilon}+1}(d_u + d_w) + \frac{2 n_1}{1 + e^{\varepsilon}} )$ \\
\bs       & $O(n_1)$ & \checkmark & $ \erroneround $ & $O( \frac{e^{\varepsilon}-1}{e^{\varepsilon}+1}(d_u + d_w) + \frac{2 n_1}{1 + e^{\varepsilon}} )$ \\
\advss  & $O(n_1)$ & \checkmark & $O(\frac{e^{\varepsilon_1}}{(1 - e^{\varepsilon_1})^2} (d_u + \frac{2e^{\varepsilon_1}}{\varepsilon_2^2}))$ & $O(\frac{e^{\varepsilon_1}-1}{e^{\varepsilon_1}+1}d_w + \frac{ n_1}{1 + e^{\varepsilon_1}} )$ \\
\advds  & $O(n)$ & \checkmark & 
\makecell{
$O(\frac{e^{\varepsilon_1}}{(1 - e^{\varepsilon_1})^2} (\alpha^2 d_u + (1-\alpha)^2 d_w)+ $ 
$\frac{2e^{2\varepsilon_1}}{(1 - e^{\varepsilon_1})^2} \frac{\alpha^2 + (1-\alpha)^2}{\varepsilon_2^2})$
\vspace{0.2mm}
} & $O(n_2 + \frac{e^{\varepsilon_1}-1}{e^{\varepsilon_1}+1} (d_w+d_u) + \frac{ 2 n_1}{1 + e^{\varepsilon_1}})$ \\
\noalign{\hrule height 1pt}
\end{tabular}
}
\begin{flushleft}
* Without loss of generality, we assume $u, w \in L(G)$. 
$\varepsilon$ is the overall privacy budget. 
In \advss and \advds, $\varepsilon_1$ represents the privacy budget for randomized responses, while $\varepsilon_2$ denotes the budget for the Laplace mechanism. 
$\alpha \in [0,1]$ adjusts the contribution of $\fq$ and $\fx$. 
$d_u$ and $d_w$ represent the degrees of $u$ and $w$ in $G$. 
$n_1 = |U(G)|$, $n_2 = |L(G)|$, and $n = |U(G) \cup L(G)|$. 
% $n_1$ and $n_2$ denotes the number of vertices in $U(G)$ and $L(G)$. 
% $n$ denotes the number of all vertices in $G$. 
\end{flushleft}
        \myspace
        \myspace
\label{tab:complexitycompare}
\end{table*}

\section{Experimental evaluation}
\label{sec:exp}

In this section, we evaluate the proposed common neighbor estimation algorithms under \epldp through experiments. 

\subsection{Experimental Settings}
\begin{figure}[ht]
\centering
\subfigure[The mean absolute error]{ 
\includegraphics[width=0.5\textwidth]{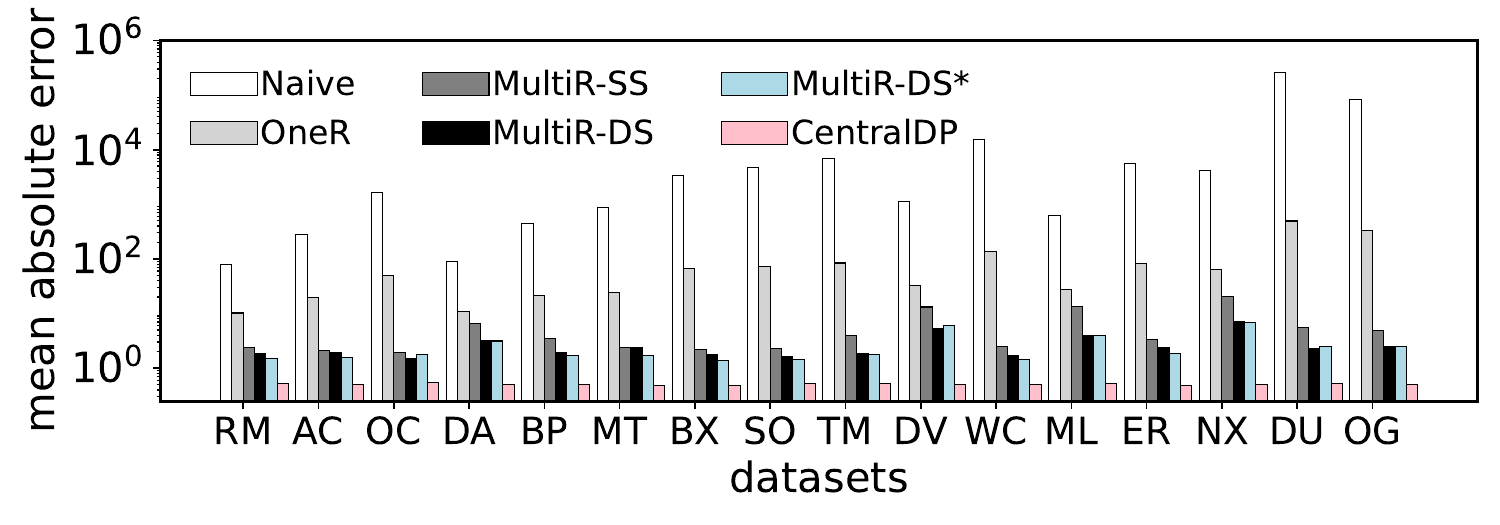}
\label{default:effect}
}
\subfigure[The computational time cost
% \textcolor{red}{
% TODO: add the running time of \advdeg algorithms, it should be faster than \advds. 
% also, mention that the degree estimation time should be amortized over many query pairs. 
% Question: can we improve the efficiency of \advds without the additional assumption}
]{
\includegraphics[width=0.5\textwidth]{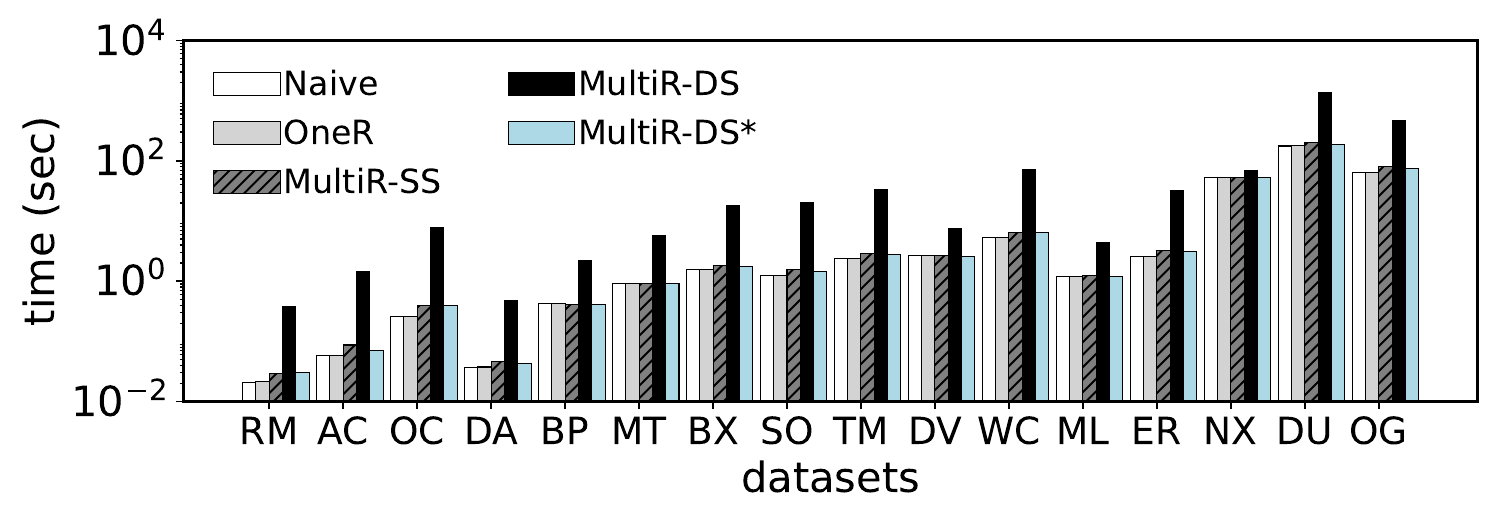}
\label{default:time}
}
\caption{Performance on different datasets ($\epsilon = 2$) }
\myspace
\label{fig.default}
\end{figure}

\begin{figure*}[thb]
\centering
    \begin{minipage}{0.66\textwidth}
    \centering
        \subfigure[\texttt{Stackoverflow}]{
        \includegraphics[width=\figsizeone\textwidth]{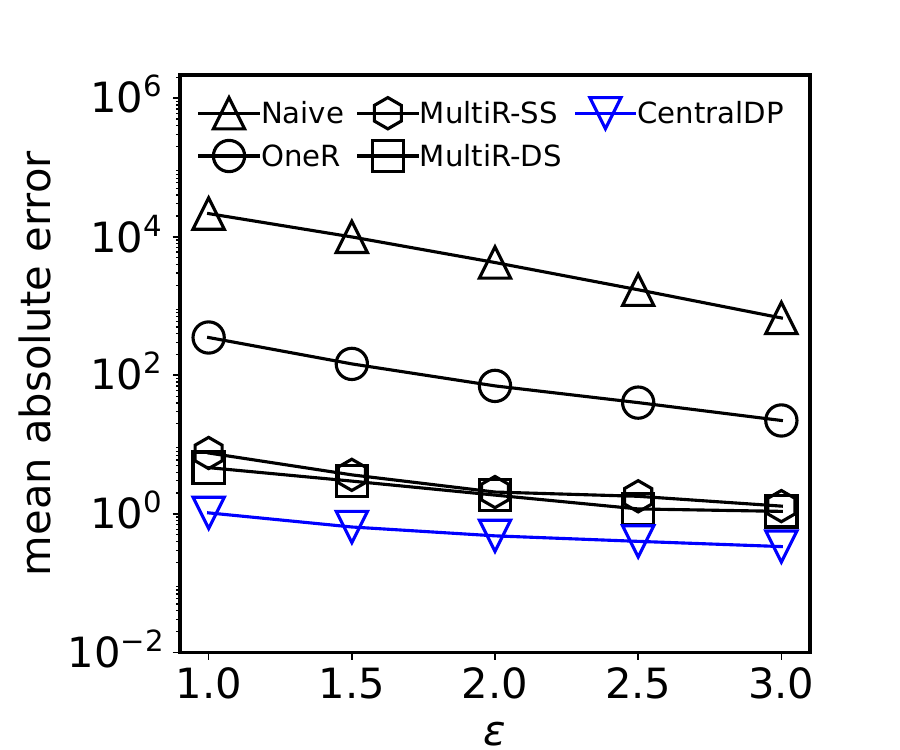}
        }
        \subfigure[\texttt{Team}]{
        \includegraphics[width=\figsizeone\textwidth]{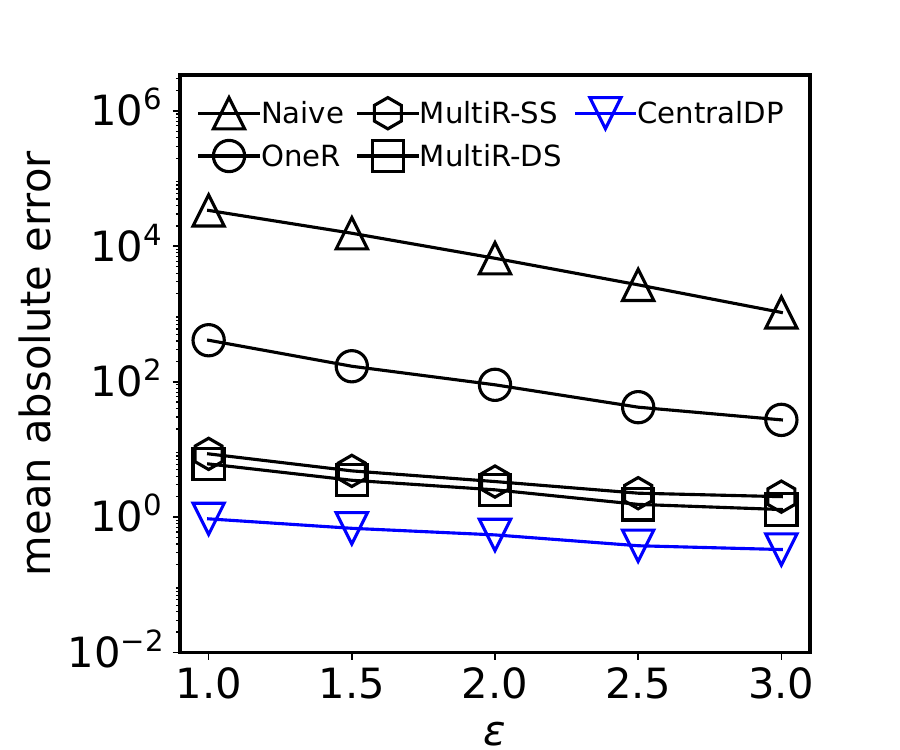}
        }
        \subfigure[\texttt{Wiki-En-Cat}]{
        \includegraphics[width=\figsizeone\textwidth]{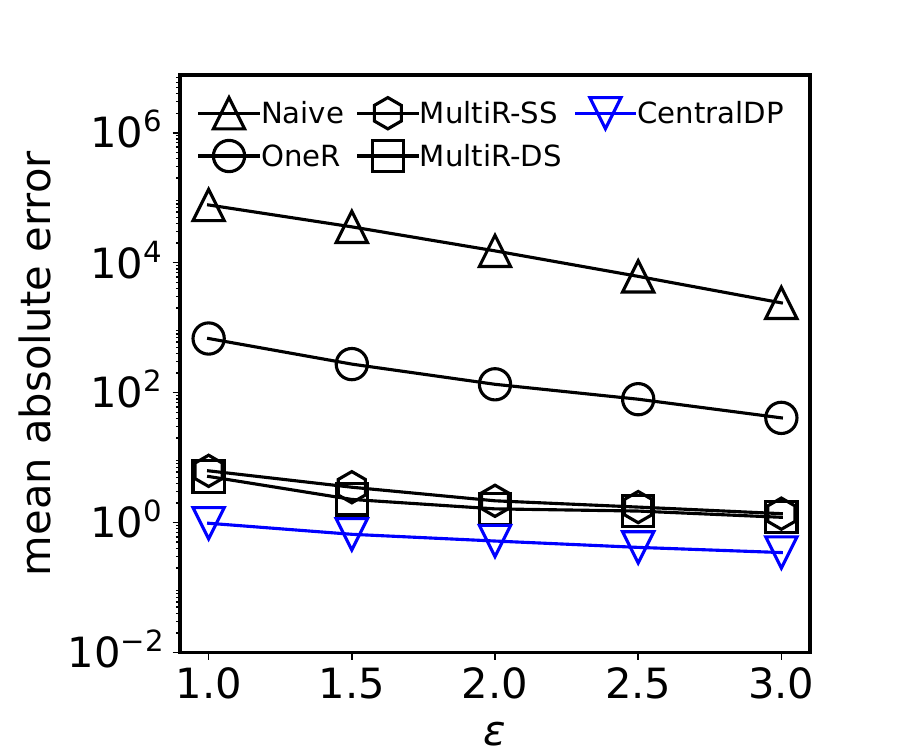}
        }
        \subfigure[\texttt{Movielens}]{
        \includegraphics[width=\figsizeone\textwidth]{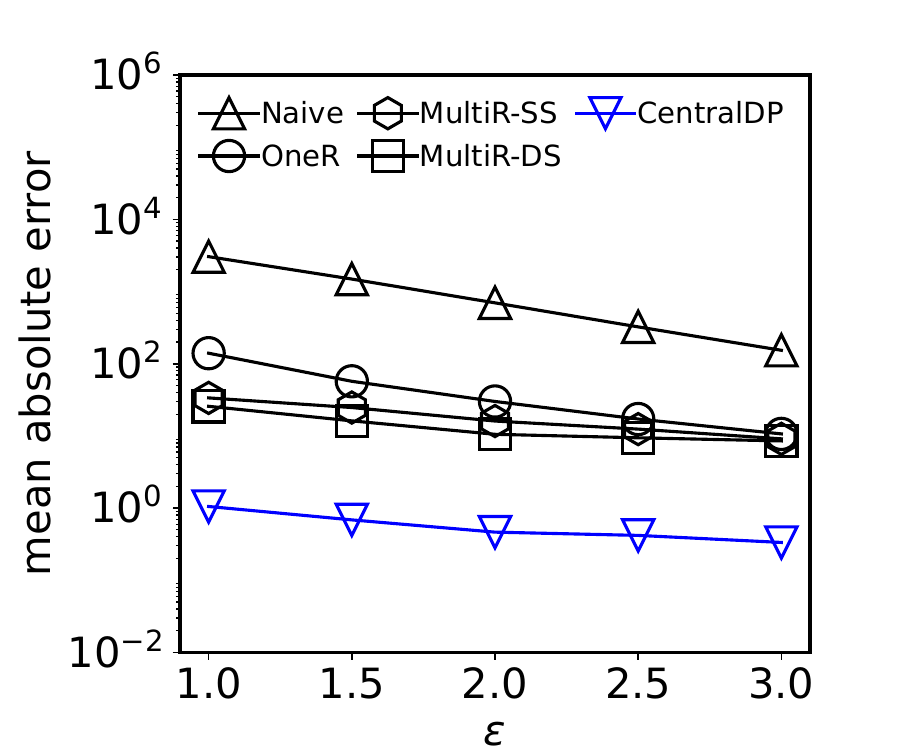}
        }
        \myspace
        \subfigure[\texttt{Epinions}]{
        \includegraphics[width=\figsizeone\textwidth]{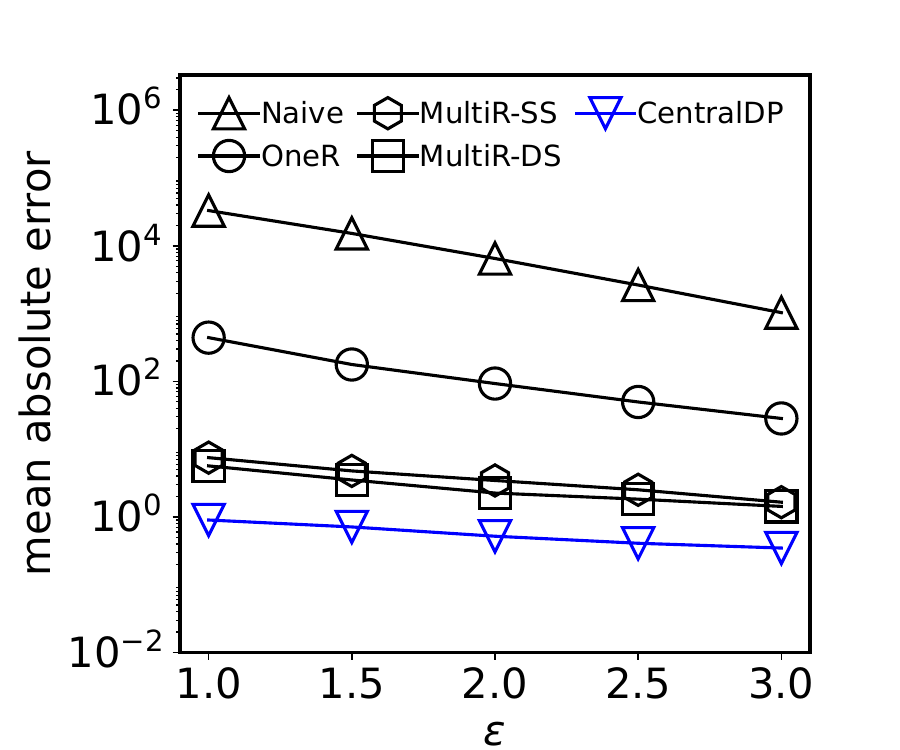}
        }
        \subfigure[\texttt{Netflix}]{
        \includegraphics[width=\figsizeone\textwidth]{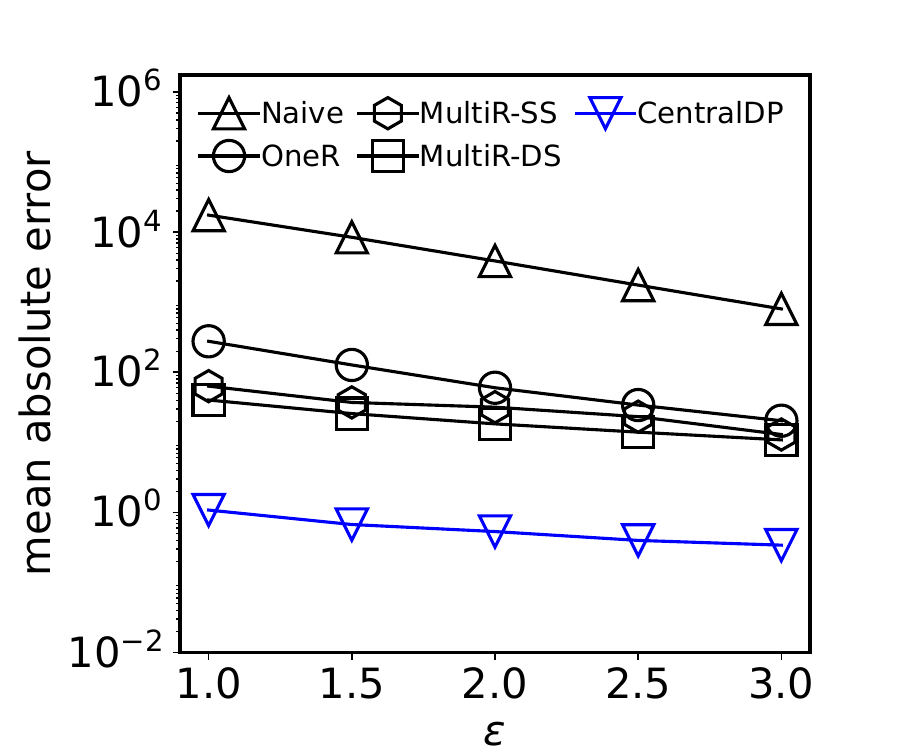}
        }
        \subfigure[\texttt{Delicious-ui}]{
        \includegraphics[width=\figsizeone\textwidth]{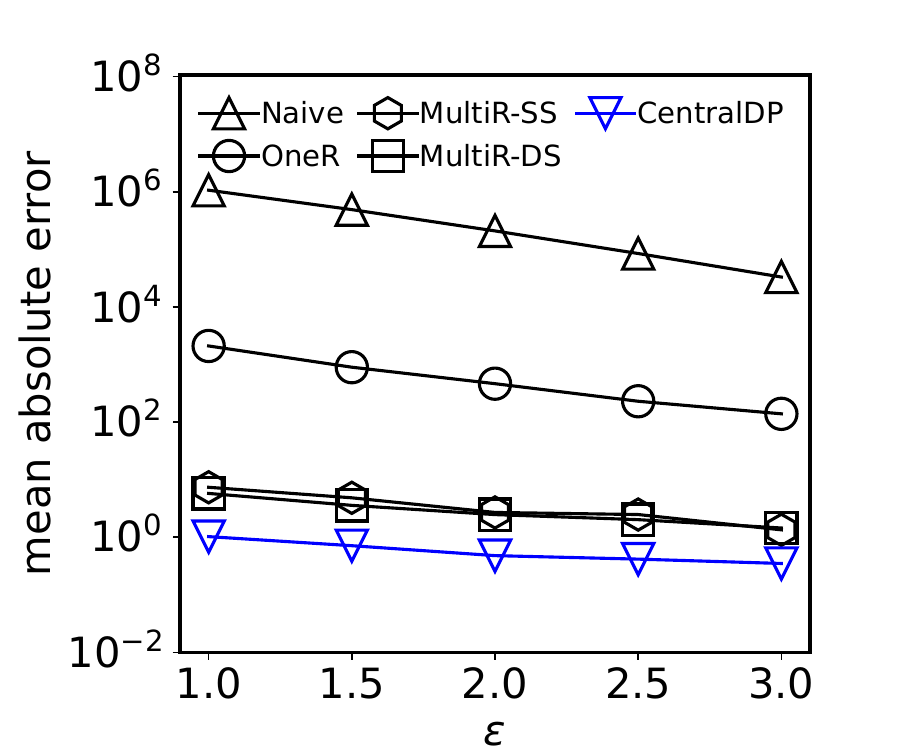}
        }
        \subfigure[\texttt{Orkut}]{
        \includegraphics[width=\figsizeone\textwidth]{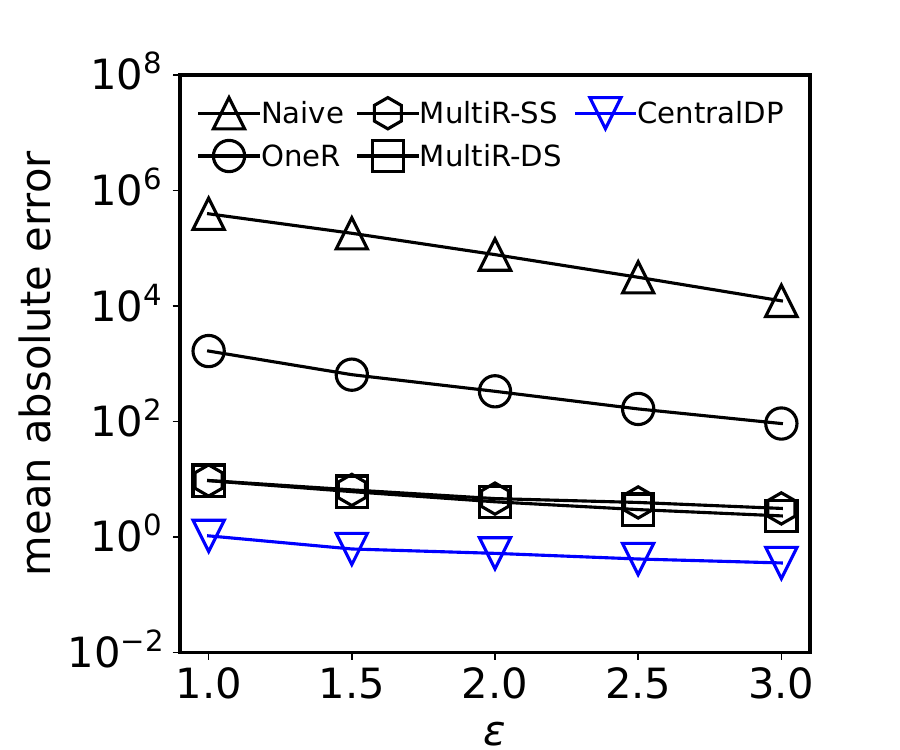}
        }
    \myspace
    \caption{{Effect of $\varepsilon$ on the mean relative errors of \naive, \bs, \\ and \advds.}}
    \label{fig.vary}
    \end{minipage}
\hfill
    \begin{minipage}{0.33\textwidth}
    \centering
        \subfigure[ \texttt{Team}]{
        \includegraphics[width=\figsizetwo\textwidth]{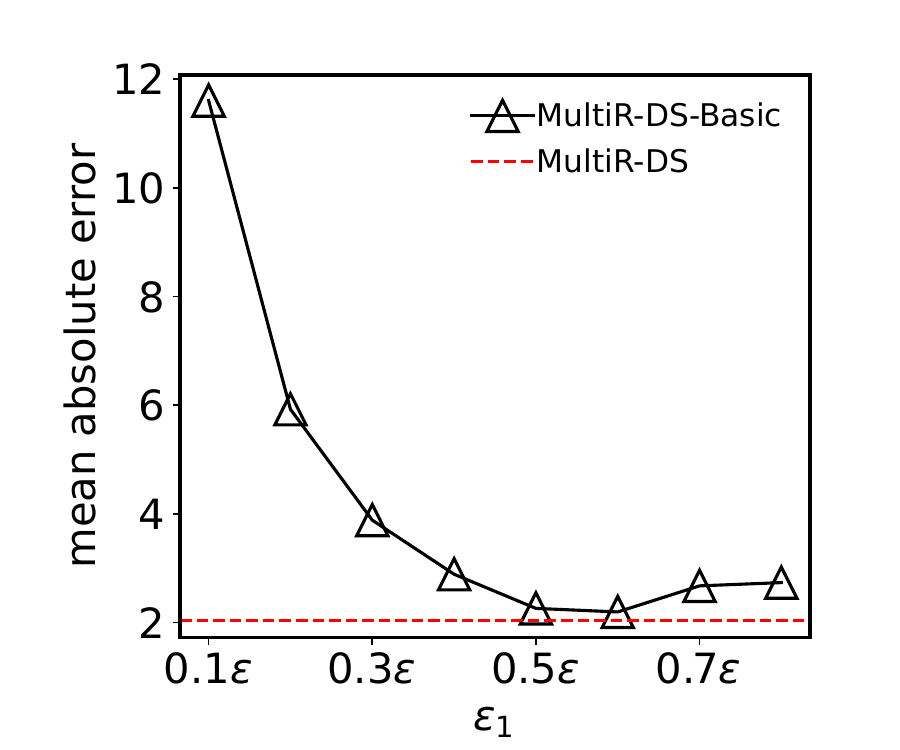}}
        \subfigure[ \texttt{Bookcrossing}]{
        \includegraphics[width=\figsizetwo\textwidth]{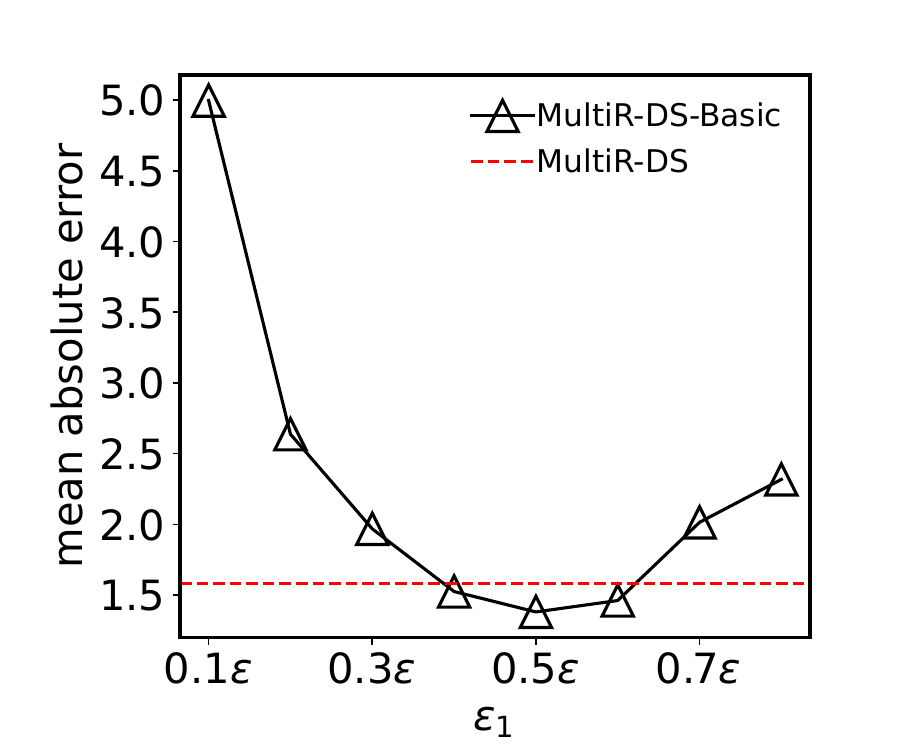}}
        \myspace
        \subfigure[ \texttt{Delicious}]{
        \includegraphics[width=\figsizetwo\textwidth]{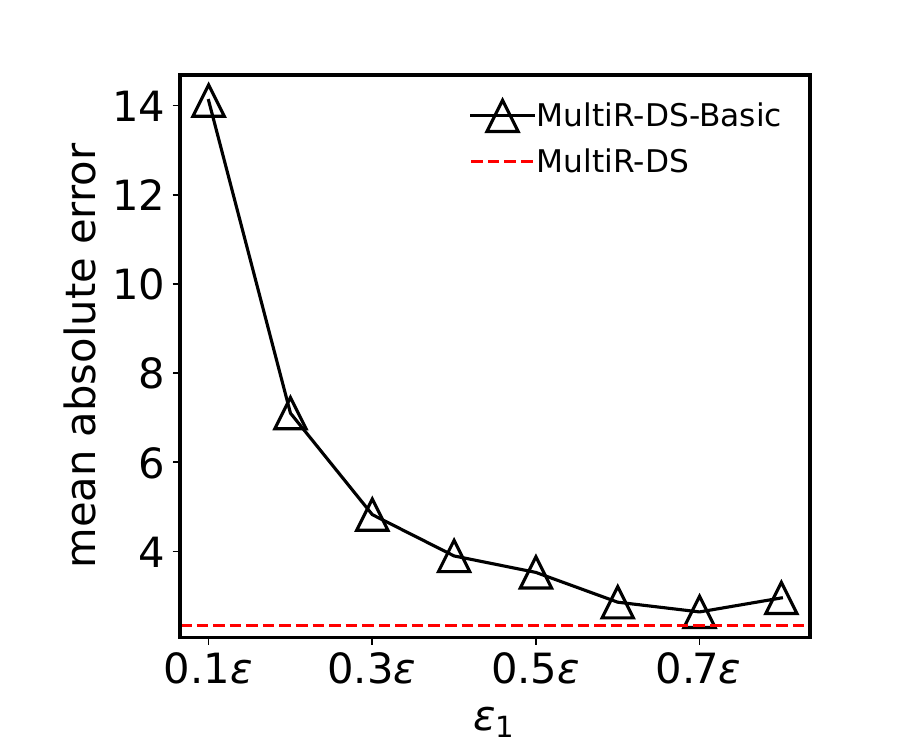}}
        \subfigure[ \texttt{Orkut}]{
        \includegraphics[width=\figsizetwo\textwidth]{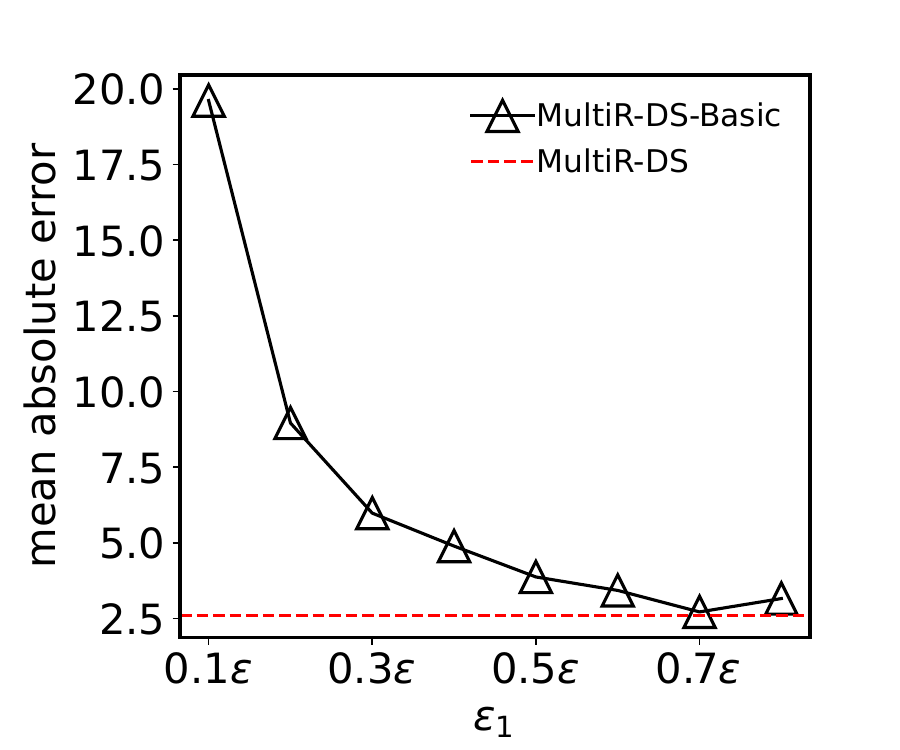}}
        \myspace
    \caption{Effectiveness of privacy budget allocation optimization.}
    \label{Fig.find}
    \end{minipage}
\end{figure*}

\begin{figure*}[thb]
\centering
    \begin{minipage}{0.33\textwidth}
    \centering
    \subfigure[ \texttt{Team}]{
    \includegraphics[width=\figsizetwo\textwidth]{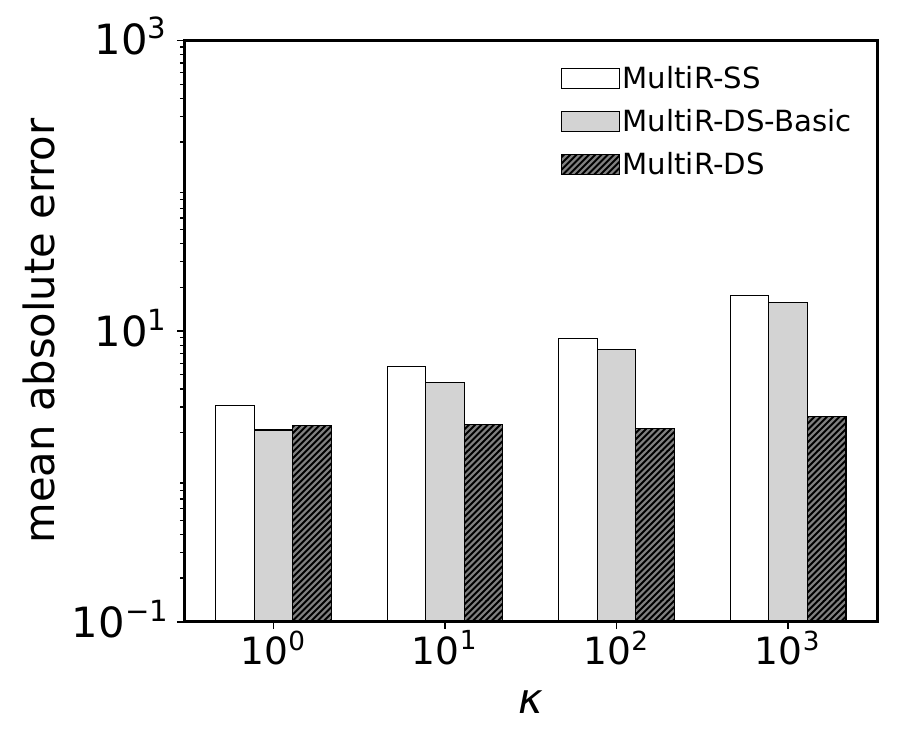}}
    \subfigure[ \texttt{Bookcrossing}]{
    \includegraphics[width=\figsizetwo\textwidth]{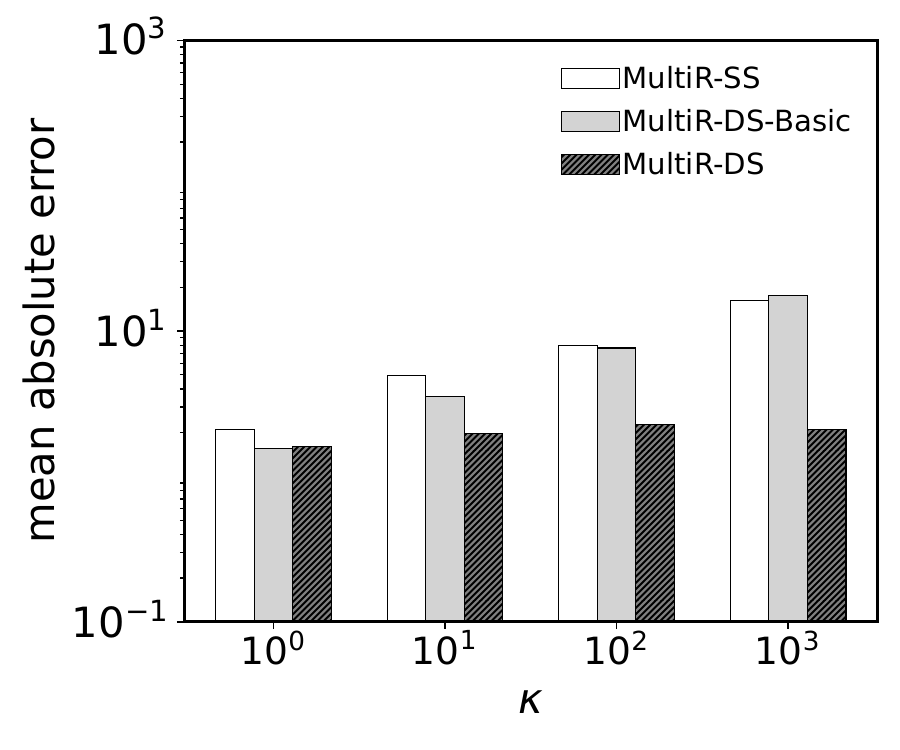}}
    % \subfigure[ \texttt{Wiki-En-Cat}]{
    % \includegraphics[width=\figsizetwo\textwidth]{figures/wiki-en-cat-adv-comparison.pdf}}
    % \subfigure[ \texttt{Epinions}]{
    % \includegraphics[width=\figsizetwo\textwidth]{figures/epinions-rating-adv-comparison.pdf}}
        \myspace
    \subfigure[ \texttt{Delicious}]{
    \includegraphics[width=\figsizetwo\textwidth]{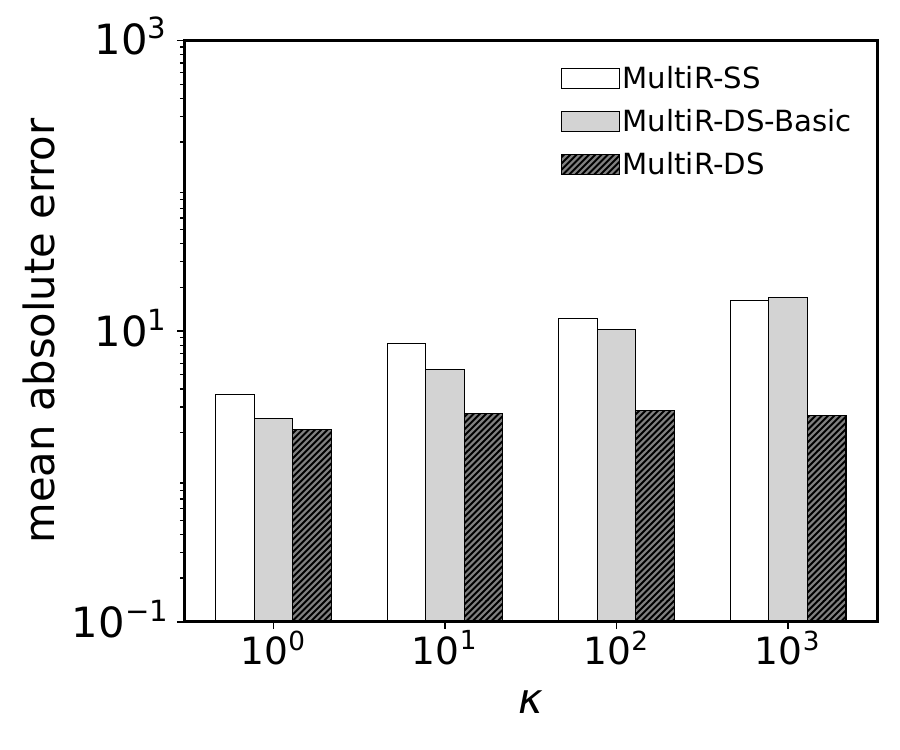}}
    \subfigure[ \texttt{Orkut}]{
    \includegraphics[width=\figsizetwo\textwidth]{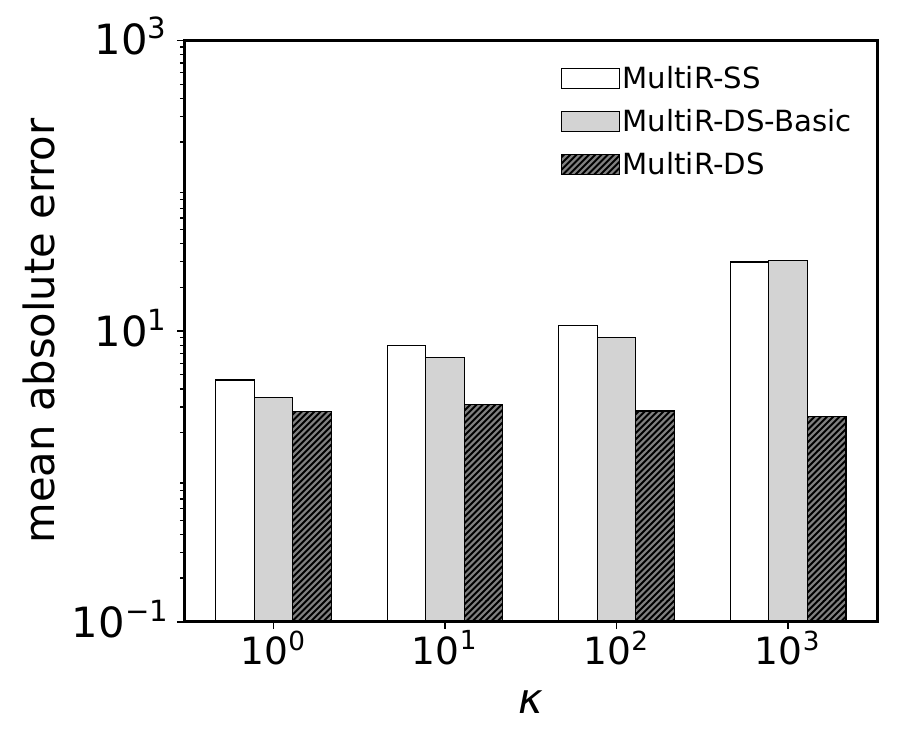}}
        \myspace
    \caption{Effectiveness of \advds. }
    \label{fig.balance}
    \end{minipage}
\hfill
    \begin{minipage}{0.33\textwidth}
    \centering
    \subfigure[\texttt{Wiki-En-Cat}]{
    \includegraphics[width=\figsizetwo\textwidth]{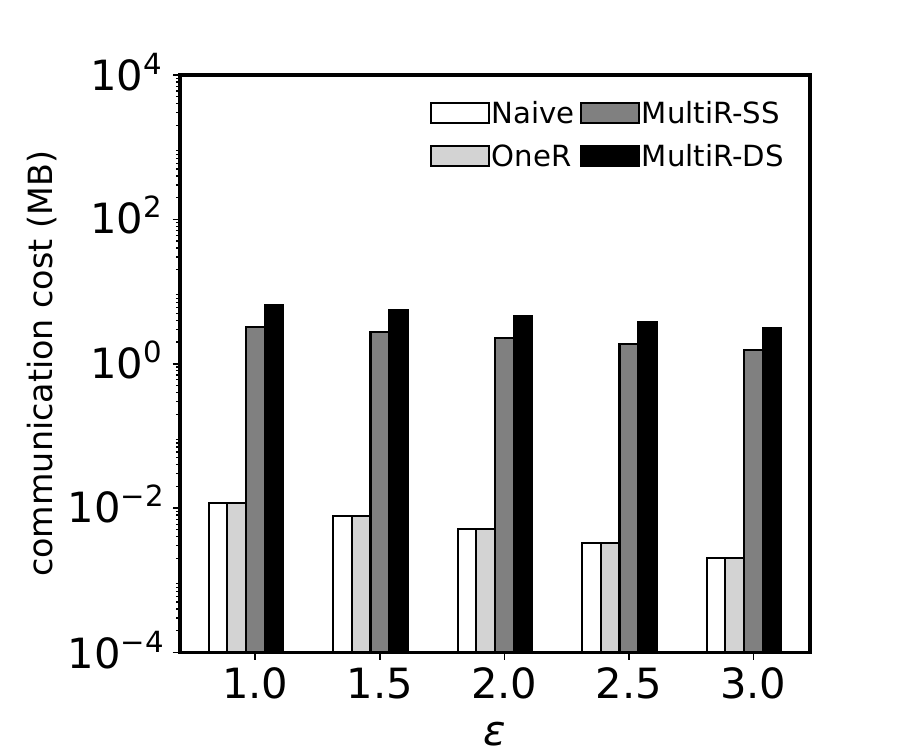}}
    \subfigure[\texttt{Epinions}]{
    \includegraphics[width=\figsizetwo\textwidth]{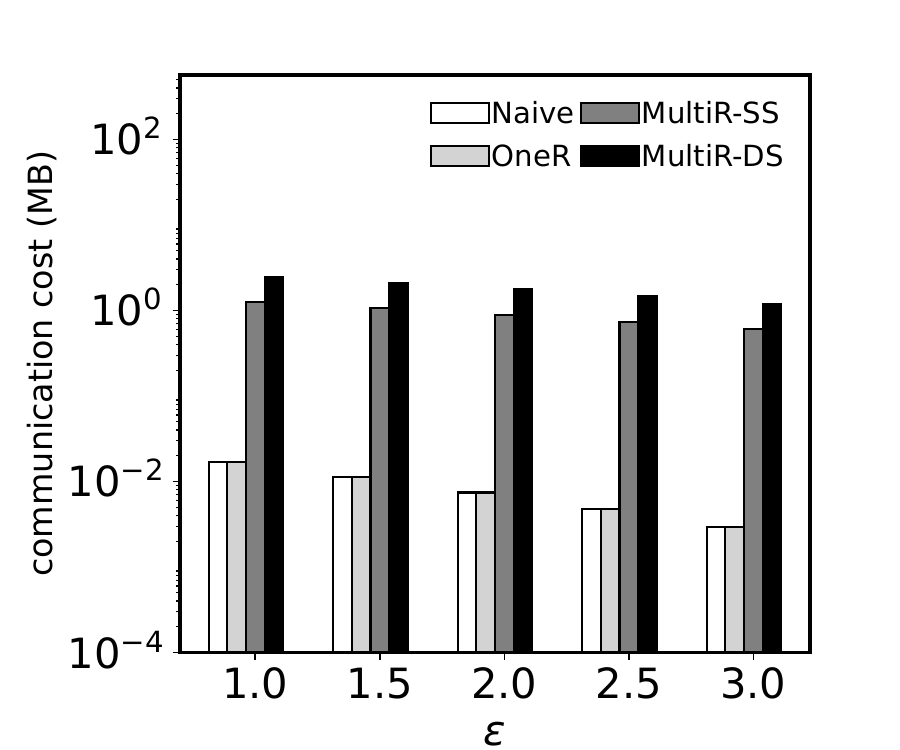}}
        \myspace
    \subfigure[\texttt{Delicious}]{
    \includegraphics[width=\figsizetwo\textwidth]{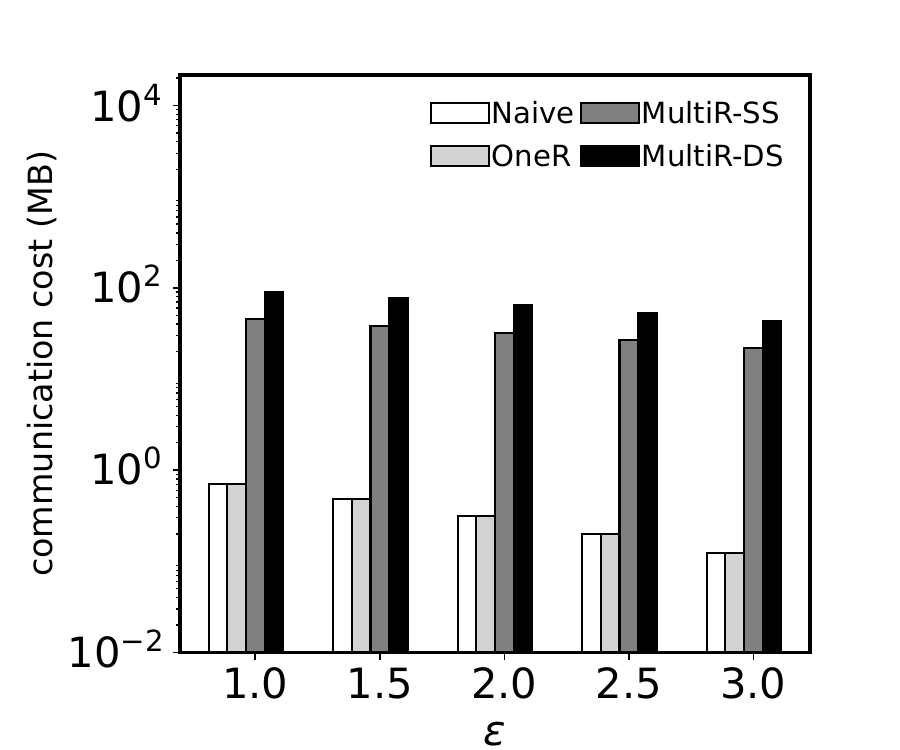}}
    \subfigure[\texttt{Orkut}]{
    \includegraphics[width=\figsizetwo\textwidth]{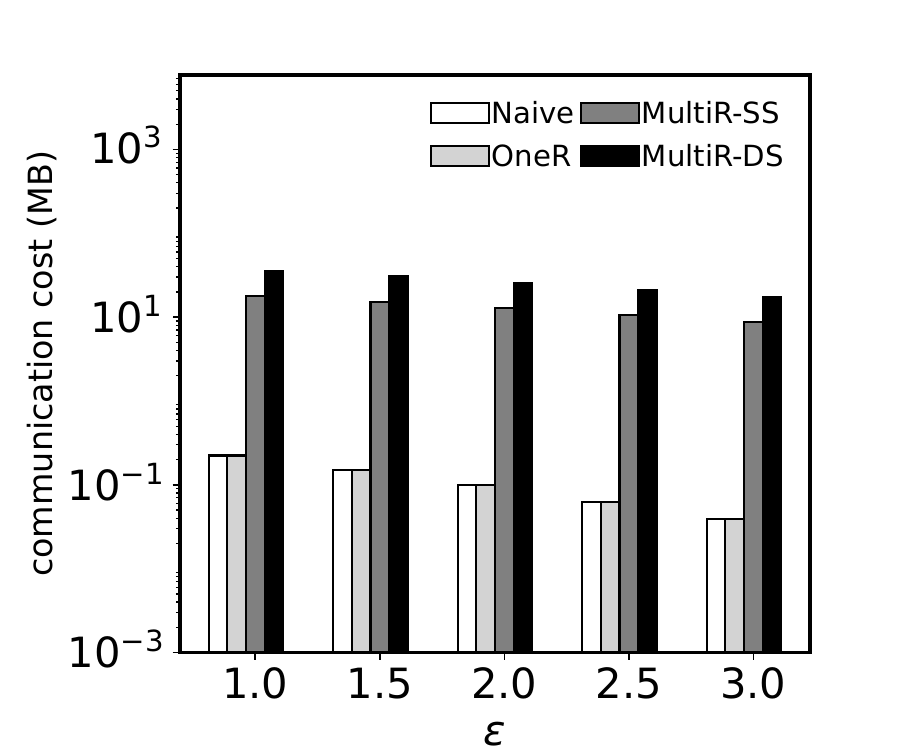}}
        \myspace
    \caption{Communication costs.}
    \label{fig.cost}
    \end{minipage}
\hfill
    \begin{minipage}{0.33\textwidth}
        \subfigure[\texttt{Wiki-En-Cat}]{
        \includegraphics[width=\figsizetwo\textwidth]{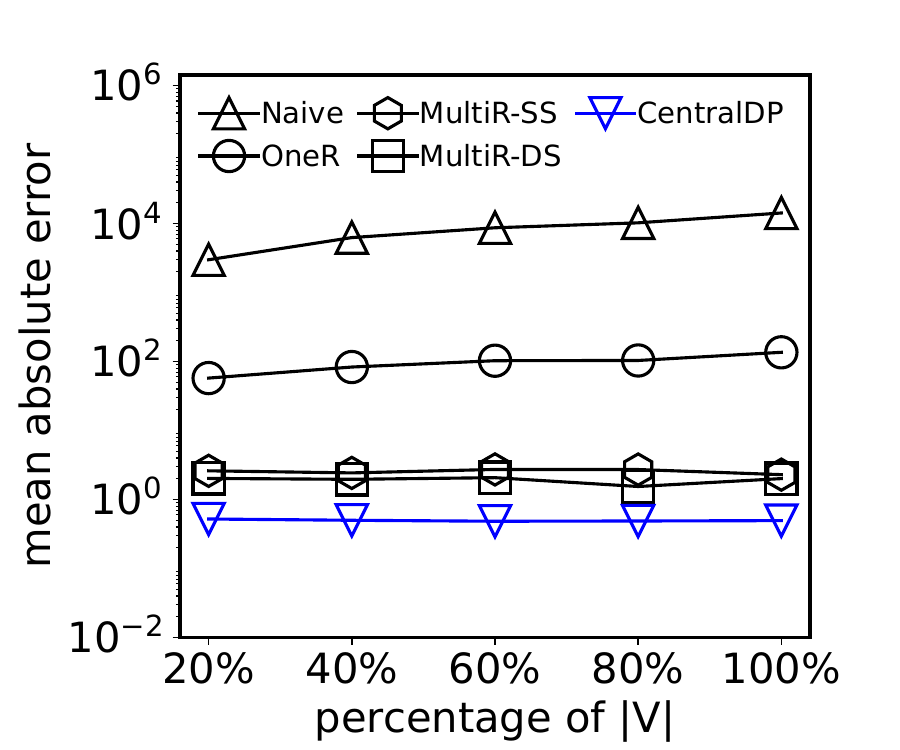}}
        \subfigure[\texttt{Epinions}]{
        \includegraphics[width=\figsizetwo\textwidth]{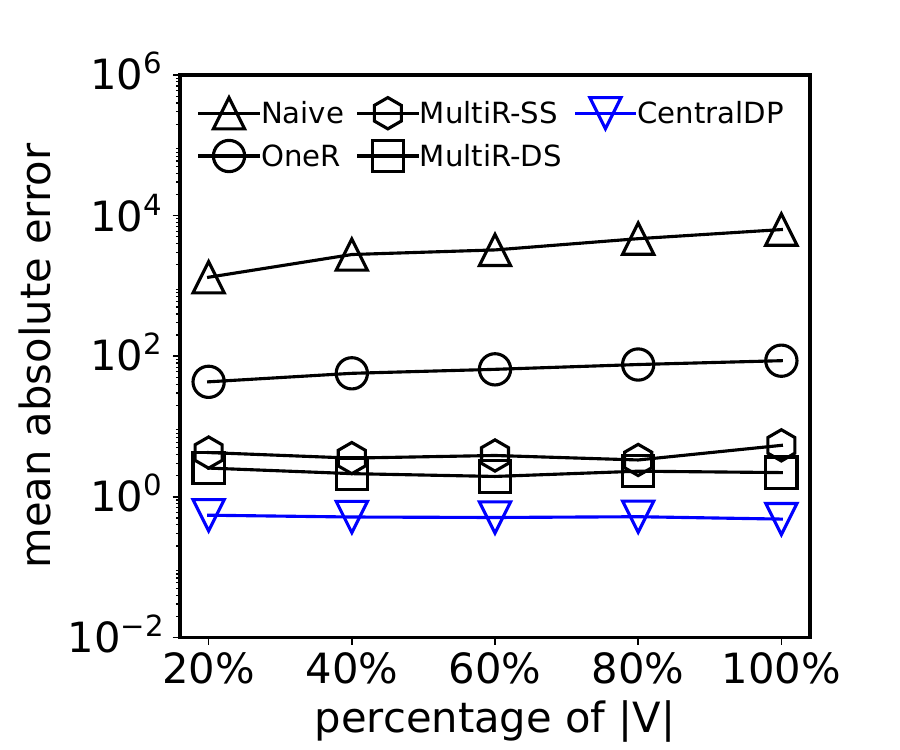}}
        \myspace
        \subfigure[\texttt{Delicious}]{
        \includegraphics[width=\figsizetwo\textwidth]{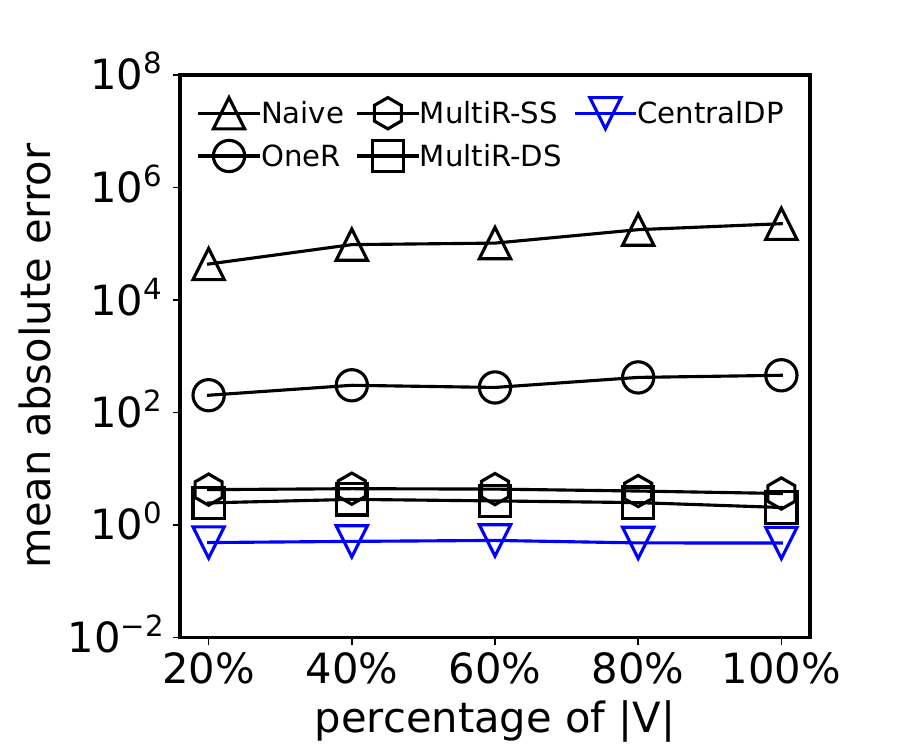}}
        \subfigure[\texttt{Orkut}]{
        \includegraphics[width=\figsizetwo\textwidth]{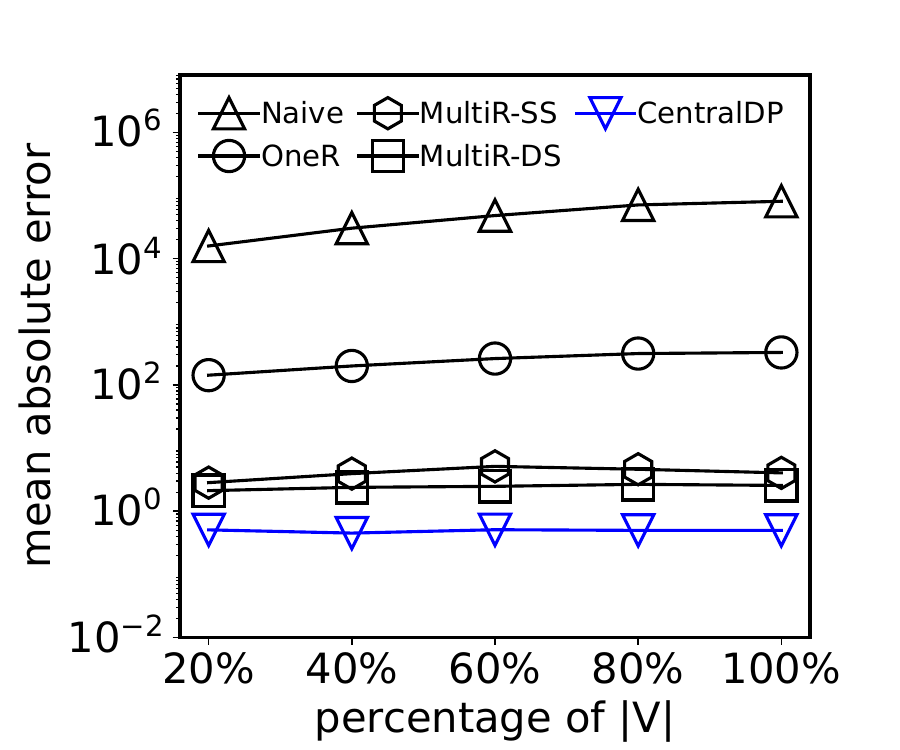}}
        \myspace
        \caption{Effect of the number of vertices.}
        \label{fig.scale}
    \end{minipage}
\end{figure*}

\noindent
{\bf Datasets.}
We use 15 datasets from KONECT ({\url{http://konect.cc/}}). 
Table \ref{tab:summary} shows the statistics of the datasets. 
$|U|$ and $|L|$ are the number of vertices in the upper and lower layers. $|E|$ is the number of edges in the graph. 

\noindent
{\bf Algorithms.} 
We evaluate the following common neighborhood estimation algorithms under \epldp. 
\begin{itemize}
    \item \naive: the algorithm that returns the number of common neighbors between $u$ and $w$ on the noisy graph $G'$;
    \item \bs: the one-round algorithm returns an unbiased estimate of $\pqx$ based on the noisy graph $G'$;
    \item \advss: the multiple-round single-source algorithm that returns the unbiased estimator $\fq(u,w)$ by utilizing the local neighborhood of $u$;     
    \item \advds: the multiple-round double-source algorithm that returns the unbiased estimator $ \alpha \fq + (1 - \alpha) \fx$ ($\alpha \in [0,1]$) by utilizing the local neighborhoods of both $u$ and $w$.
\end{itemize}
{We also implement two variants of \advds: \advdsbasic and \advdeg. 
\advdsbasic returns the average of the two single-source estimators $\fq$ and $\fx$ ($ \frac{\fq +\fx}{2}$). It spends $\varepsilon_1$ on noisy graph construction and $1 - \varepsilon_1$ on the Laplace mechanism and does not estimate $deg(u)$ or $deg(w)$. 
Similarly to \advds, \advdeg returns $ \alpha \fq + (1 - \alpha) \fx$ ($\alpha \in [0,1]$) and adopts the same optimizations for the allocation of privacy budgets to find $\varepsilon_1$ and $\alpha$. 
The difference is that \advdeg assumes that the vertex degrees are public and it does not need an additional round for vertex degree estimation. }

{To better evaluate the edge LDP algorithms, we also implement the \cdp algorithm under the centralized model, which assumes the data curator has access to the entire bipartite graph. \cdp directly applies the Laplacian mechanism to the number of common neighbors of two query vertices. Since the global sensitivity of $\pqx$ in the central model is $1$, \cdp returns $\pqx + Lap(\frac{1}{\varepsilon})$. }
All algorithms are implemented in C++. The experiments are run on a Linux server with an Intel Xeon 6342 processor and 512GB memory.
% and let it run on varying $\varepsilon_1$ values

\noindent
{\bf Parameter settings.} 
By default, the privacy budget $\varepsilon$ is set to $2$. We also allow it to vary from $1$ to $3$ with increments of $0.5$. 
For \advss and \advdsbasic, $\varepsilon_1$ is set to $0.5 \varepsilon$ by default. 
For \advds, we set $\varepsilon_0 = 0.05 \varepsilon$ for degree estimations. 
For each algorithm, we uniformly sample $100$ vertex pairs on the same layer and report the mean absolute error, the average of the absolute differences between the predicted and true values across all sampled vertex pairs. 
To evaluate the performance of \advds, we use $\kappa$ to quantify the imbalance between two vertex degrees. 
Specifically, on a given pair of vertices ($u$ and $w$) with the parameter $\kappa$,
we have $\max(deg(u), deg(w)) > \kappa \times \min(deg(u), deg(w))$. 
% either either 
% $deg(u) > \kappa deg(v)$ or $deg(v) > \kappa deg(u)$.
% Specifically, in a sample of vertex pairs with $\kappa$, each pair $u$, $v$ exhibits one of two scenarios: either $deg(u) > \kappa deg(v)$ or $deg(v) > \kappa deg(u)$. 
% For each algorithm, we run it for $10$ rounds and report the average relative error. 

\noindent
\textbf{Evaluate the effectiveness of edge LDP algorithms across different datasets.} 
In Fig. \ref{default:effect}, we report the performances of the edge LDP algorithms including \naive, \bs, \advss, \advds, and \advdeg on $100$ uniformly sampled vertex pairs when $\varepsilon = 2$. 
{Note that we also include the performance of \cdp under the centralized model for comparison.} 
% Overall trend: 
First, we observe that the multiple-round algorithms (\advss, \advds, and \advdeg) significantly outperform \naive and \bs across all datasets. 
% Specifically, \advss and \advds significantly outperform \naive and \bs, achieving mean absolute errors lower by up to four and two orders of magnitude, respectively.
Specifically, \advss and \advds achieve mean absolute errors lower by up to four and two orders of magnitude, respectively. 
% by up to four and two orders of magnitude in mean absolute errors, respectively. 
% Specifically, \advss and \advds demonstrate significantly better performance than \naive and \bs, achieving mean absolute errors lower by up to four and two orders of magnitude, respectively.
% reason:
This is because \advss and \advds address the overcounting issue due to the dense noisy graph with \naive by deriving unbiased estimates. 
Meanwhile, compared to the \bs algorithm that produces unbiased estimates by considering all vertices on the opposite layers as possible common neighbors, \advss and \advds induce much smaller mean absolute errors by reducing the candidate pool common neighbors to the neighbors of query vertices. 
% pattern 2: \advds > advss
We also observe that \advds consistently produces smaller mean absolute errors than \advss. 
For example, on \texttt{Netflix}, the mean absolute error of \advds is approximately one-fifth that of \advss. 
% based on the local neighborhoods of both query vertices 
This is because \advds integrates the two single-source estimators and dynamically adjusts the privacy budget allocation based on the query vertices. 
%%% pattern new: advdeg better than advds. 
{
In \figuresixa, we observe that \advdeg generally produces slightly smaller mean absolute errors compared to \advds. 
%%% reason: 
This is because \advdeg does not need to spend an additional privacy budget for degree estimation, which leads to more privacy budgets for noisy graph construction and the Laplace mechanism (i.e., $\varepsilon_1$ and $\varepsilon_2$ becomes larger). }
% pattern 3: between \bs and \naive, with statistics.
We also observe that \bs achieves much lower mean absolute errors than \naive because \bs leverages flipping probability to obtain unbiased estimators, which mitigates the over-counting issue in \naive. 

%%% additional pattern: central is the best. 
{
In addition, \cdp results in lower errors than all algorithms with edge LDP. 
This illustrates the limitations of the local model in terms of data utility due to stronger privacy guarantees. 
}

\noindent
\textbf{Evaluate the computational time costs across datasets.} 
{
In Fig. \ref{default:time}, we report the computational time costs of \naive, \bs, \advss, \advds, and \advdeg on $100$ vertex pairs when $\varepsilon = 2$. }
Note that our evaluation focuses on the computational time costs incurred by both the vertex side and the data curator side. 
We can observe that the time costs of \naive, \bs, and \advss remain relatively comparable while \advds requires more time. 
This is because the time complexities of \naive, \bs, and \advss depend on the number of vertices on the opposite layers of the query vertices i.e., $O(n_1)$, which is incurred by noisy graph construction. 
Since \advds needs an additional $O( n_2)$ time for the estimation of the average degree, its total time complexity becomes $O(n)$ and exceeds the others. 
% its dominating time complexity is $O(|V(G)|)$, which arises from estimating the average degree of vertices sharing the same layer as the query vertices. 
Despite this, \advds remains highly efficient and can scale effectively to bipartite graphs with 300 million edges (i.e., \texttt{Orkut}). Also, in practice, the time required for average degree estimation is distributed across vertices. 
{Additionally, we observe that the \advdeg algorithm runs faster than \advds and incurs comparable time costs to \advss. This is because \advdeg does not need an additional round to estimate the vertex degrees. }

\noindent
\textbf{Evaluate the effect of the privacy budget $\varepsilon$.} 
As shown in Fig.~\ref{fig.vary}, we report the mean absolute errors of \naive, \bs, and \advds on 8 datasets, as $\varepsilon$ varies. 
{Note that we also include \cdp under central differential privacy for comparison.} 
% trend1: all become better as ep increases 
We observe that all algorithms produce smaller mean absolute errors as $\varepsilon$ increases, which is consistent with our L2 loss analysis. 
As $\varepsilon$ increases, the difference between any noisy graph constructed by randomized responses and the input graph becomes smaller. 
% For \naive, this means directly counting the number of common neighbors on the noisy graph will be more accurate. 
% for bs, 
Another pattern is that the multiple-round algorithms (\advss and \advds) significantly outperform \naive and \bs, with mean absolute errors up to four orders of magnitude lower. 
{This is because the expected L2 losses of \naive and \bs are $O(n_1^2)$ and $O(n_1)$, respectively, while the expected L2 losses of \advss and \advds only depend on vertex degrees.} 
\advds consistently outperforms \advss on varying values of $\varepsilon$ because \advds integrates both single-source estimators and employs privacy budget allocation optimizations for minimized L2 loss. 
We also observe that \bs consistently outperforms \naive as the privacy budget increases. 
{
As expected, \cdp produces smaller mean absolute errors than algorithms under edge LDP, which has stronger privacy guarantees. 
}

% This is because \bs leverages flipping probability during noisy graph construction to obtain unbiased estimators, which mitigates the over-counting issue in \naive.

\noindent
\textbf{Evaluate the effect of privacy budget allocation optimization on \advds. }
In Fig.~\ref{Fig.find}, we present the mean absolute errors of \advdsbasic in four datasets, as $\varepsilon_1$ ranges from $0.1 \varepsilon$ to $0.7 \varepsilon$, where $\varepsilon = 2$  and $\varepsilon_2 = \varepsilon - \varepsilon_1$. 
Note that \advdsbasic does not employ the privacy budget allocation optimization. 
% and its performance is based on a predetermined split between $\varepsilon_1$ and $\varepsilon_2$. 
In contrast, \advds adjusts $\varepsilon_1$ and the contribution of two single-source estimators (measured by $\alpha$) based on the query vertices. 
We use red dashed horizontal lines to indicate the mean absolute errors associated with \advds. 
First, the optimal budget allocation plan varies across datasets and it is unrealistic to fix $\varepsilon_1$ and $\varepsilon_2$ for all datasets. 
This is because optimal budget allocation depends on the degrees of the query vertices, as shown in Table \ref{tab:complexitycompare}. 
%%%% 
Also, for each dataset, the mean absolute error with \advds is close to or even smaller than the smallest mean absolute error of \advdsbasic on varying values of $\varepsilon_1$. 
% This implies that \advds can find the allocation of the privacy budget and the weights of two single-source estimators that result in near-optimal L2 loss. 
This implies that \advds can find $\varepsilon_1$ and $\alpha$ that result in near-optimal L2 loss. 

% {\color{red} as well as the weighting of two estimators} that approximates the optimum. 

%%% need to add more explanations to this experiment 
\noindent
\textbf{Evaluate the effectiveness of \advds on vertex pairs with imbalanced degrees.} 
In Fig.~\ref{fig.balance}, we report the mean absolute errors of \advss, \advdsbasic, and \advds across four datasets, as $\kappa$ ranges from $10^0$ to $10^3$, with $\varepsilon = 2$. 
Here, $\kappa$ quantifies the imbalance between two vertex degrees. 
% in a sample of vertex pairs. 
For \advdsbasic, $\varepsilon_1$ is set to $0.5 \varepsilon$. 
We observe that the mean absolute errors of \advss and \advdsbasic increase as $\kappa$ increases, while the performance of \advds remains relatively unchanged. 
%%% explain why advss and advds basic increases. 
This is because \advss only relies on one query vertex to construct the unbiased estimator $\fq$. 
Thus, if $deg(u,G)$ is large, the error increases accordingly, as indicated in Table \ref{tab:complexitycompare}. 
% as indicated by the expected L2 loss of \advss in Table \ref{tab:complexitycompare}. 
For \advdsbasic, it performs slightly better than \advss when $\kappa$ is small because it allows $\fq$ and $\fx$ to contribute equally (i.e., $(\fq + \fx)/2$). 
However, when the vertex degrees are highly imbalanced ($\kappa$ becomes large), the errors of \advdsbasic escalate rapidly. 
Also, neither \advss nor \advdsbasic can adjust privacy budget allocations based on the query vertices. 
% Also, \advss and \advdsbasic cannot adjust their privacy budget allocations. 
In contrast, \advds uses $\alpha$ to model the contribution of two query vertices and dynamically adjust privacy budgets to minimize L2 loss. 
(1) If the vertex degrees are large, \advds allocates more privacy budget to $\varepsilon_1$. 
(2) If the vertex degrees are imbalanced, \advds adjusts $\alpha$ so that the query vertex with a smaller degree has a greater contribution. 
% Additionally, we observe that \advds outperforms \advss in most settings. 
% This is because \advss only utilizes the local neighbors of $u$ and its performance is undermined when the degree of $u$ happens to be very big. \advss is more robust to imbalanced vertex pairs than \advss because it can allocate more privacy budget to $\varepsilon_1$ if the incoming vertex pairs have large degrees. 

\noindent
\textbf{Evaluate the communication costs of all algorithms.}
{
% All algorithms under \epldp must interact with the data curator.
In Fig.~\ref{fig.cost}, we report the communication costs (in MB) of each algorithm averaged across 100 randomly sampled vertex pairs in four datasets, as $\varepsilon$ varies. 
%%% pattern 1 
We observe that \naive and \bs require approximately the same message sizes. 
This is because \naive and \bs rely solely on randomized responses to satisfy edge LDP. Given a fixed $\varepsilon$, they apply randomized responses with the same flipping probability, resulting in the same expected number of noisy edges.
%%% pattern 2
% Another clear pattern is that 
Also, \advss and \advds incur higher communication costs than \naive and \bs, 
% For these multiple-round algorithms, their communication costs are incurred by 
which are incurred by 
(1) uploading the noisy edges to the data curator
(2) downloading the noisy edges to the query vertices 
(3) sending the estimators ($\fq$ or $\fx$) from the query vertices. 
For \advds, the communication costs are higher as it utilizes the noisy edges from both query vertices and also needs to send vertex degree estimated to the data curator. 
Note that the highest average communication cost for \advds across datasets is approximately $100$ MB, which is modest in practical terms. 
% Thus, when the download speed is 20 Mbps (recommended speed in YouTube[7]), the associated additional time cost is about 40 seconds. 
}

\noindent
\textbf{Evaluate the effect of the number of vertices.}
{
In Fig.~\ref{fig.scale}, we report the mean absolute errors of \cdp, \naive, \bs, \advss and \advds in four datasets as the number of vertices increases. 
Specifically, on each dataset, we uniformly sample $20\%$, $40\%$, $60\%$, $80\%$, and $100\%$ of all vertices and run the algorithms on the subgraphs induced by the sampled vertices. The privacy budget $\varepsilon$ is fixed at $2$. 
First, we observe that the performances of \cdp, \advss, and \advds remain relatively unchanged. 
This aligns with our analysis for \advss and \advds, where their L2 losses depend solely on the allocation of the privacy budget, the degrees of the query vertices, and $\alpha$, the weighting parameter adjusting the contribution of $\fq$ and $\fx$. 
For \cdp, its errors come only from the added Laplacian noise, which is also not related to the number of vertices in the bipartite graph. 
%%  pattern 2: 
We also observe that the mean absolute errors of \naive and \bs increase steadily as the number of vertices increases. 
%% evidence: 
For instance, on \texttt{Dui} and \texttt{OG}, the mean absolute errors of \naive increase approximately fivefold when the number of vertices increases from $20\% \times |V|$ to $100\% \times  |V|$. 
Meanwhile, the mean absolute errors of \bs show less sensitivity, increasing approximately 2.3 times. 
%% reason
This is because the expected L2 losses of \naive and \bs are $O(n_1^2)$ and $O(n_1)$, respectively. }

\section{Related Work}
\label{sec:related}
Here we review the related works on graph analysis under differential privacy. 

\noindent
{\bf Graph analysis under differential privacy.}
% DP widely applied to graphs. 
Differential privacy is widely adopted for privacy-preserving graph analysis, including releasing degree distributions {\cite{hay2009boosting, hay2009accurate, day2016publishing, macwan2018node}}, 
{common neighbor count distribution \cite{lv2024publishing}}, 
$k$-star counting \cite{nissim2007smooth, karwa2011private}, {triangle counting \cite{ ding2018privacy, imola2021locally,lv2021publishing, imola2022communication, imola2022differentially}}, and core decomposition \cite{dhulipala2022differential}, and {graph learning \cite{sajadmanesh2021locally, lin2022towards, wu2022linkteller, zhu2023blink, wei2024poincare}}. 
% Some approaches adopt {\em central differential privacy} \cite{nissim2007smooth, hay2009boosting, karwa2011private, zhang2015private, lv2024publishing}, where a trusted data curator can access the whole graph. 
% However, if the central data curator is corrupted or hacked, the privacy of all users is breached \cite{imola2021locally, jiang2021applications}.
Some approaches adopt {\em central differential privacy} \cite{nissim2007smooth, hay2009boosting, karwa2011private, zhang2015private, lv2024publishing}, where a trusted curator can access the entire graph. However, if this curator is compromised, all users' privacy is at risk \cite{imola2021locally, jiang2021applications}. 
% central models face challenges in finding a trusted third party. 
%%% Local model
{Another line of research adopts {\em local differential privacy} \cite{qin_generating_2017, wei2020asgldp, imola2021locally, imola2022communication, eden2023triangle,liu2022collecting,  sun2019analyzing, imola2022differentially}. }
Two main paradigms exist for graph analysis under LDP: 
(1) {general-purpose synthetic graph construction \cite{qin2017generating, zhang2018two, gao2018local, ju2019generating, liu2020privag, ye2020lf,hou2023ppdu}} and 
(2) problem-specific algorithmic design. 
The former often suffers from low data utility due to the loss of graph structure. 
Under the second category, many works are devoted to motif counting. 
%%% triangle counting 
\cite{imola2021locally} introduces one-round and two-round algorithms for triangle counting under edge LDP, while \cite{imola2022communication} improves communication cost and estimation error. 
\cite{eden2023triangle} offers an in-depth technical analysis of these algorithms. 
\cite{liu2022collecting} and \cite{sun2019analyzing}  study triangle counting in the localized setting with extended local views. 
%%% crypto assisted: 
{\cite{liu2023cargo} attempts to improve data utility for triangle counting under edge LDP in a crypto-assisted manner.} 
\cite{sun2019analyzing} also addresses three-hop paths and k-cliques on small $k$ values. 
%%% cycle counting under the shuffle model. 
\cite{imola2022differentially} estimates the 4-cycle and triangle counts under the shuffle model, where users' messages are shuffled before being sent to the data curator. 
% In a recent study \cite{imola2022differentially}, the 4-cycle and triangle counts are estimated using the shuffle model, wherein users' messages undergo shuffling before reaching the data curator.
% This model relies on the additional assumption that the data curator and the shuffler do not collude. 
% The loss is purely additive, and the additive losses match or improve upon the best-known previous additive loss in any version of differential privacy when 1/δ is polynomial in n
% derives approximate algorithms for 
% In addition, there are also works devoted to core decomposition, degree 
\cite{sun2024k} proposes k-star LDP to addresses differentially private $(p,q)$-biclique counting over bipartite graphs. Specifically, each vertex reports its perturbed k-star neighbor lists instead of the classic edge neighbor lists to the data curator. 
In addition, \cite{dhulipala2022differential} studies core decomposition under edge LDP, leading to approximate solutions for densest subgraph discovery. 
\cite{dinitz2023improved} further proves a purely additive loss for the densest subgraph problem under edge LDP. 
{A recent work \cite{lv2024publishing} studies publishing the histogram of common neighbor counts under the centralized model, which differs from our setting. }

\section{Conclusion}
\label{sec:conclusion}
In this paper, we study the problem of common neighborhood estimation on bipartite graphs under edge LDP. 
To address overcounting with the \naive approach, we propose the \bs algorithm that leverages the flipping probability to construct unbiased estimates. 
% \bs considers all vertices on the opposite layer as common neighbor candidates. 
To improve data utility, we propose a multiple-round framework and a single-source algorithm \advss, which enables the query vertices to download noisy edges and construct unbiased estimators locally. 
This significantly reduces error compared to \bs by limiting the candidate pool. 
% for common neighbors to the neighbors of query vertices. 
To tackle cross-round privacy budget allocation and the variety of query vertices, 
% To tackle cross-round privacy budget allocation and various queries, 
we propose the \advds algorithm that returns a weighted average of two unbiased estimators associated with two query vertices. 
We propose novel optimizations to adjust the privacy budgets of each round and the contribution of each estimator based on the query vertices. 
Experiments on $15$ real-world bipartite graphs validate the effectiveness and efficiency of the proposed multiple round algorithms.

% \section{Acknowledgment}
% Kai Wang is supported by NSFC 62302294 and NSFC U2241211. 
% Wenjie Zhang is supported by ARC DP230101445 and FT210100303. 
% Xuemin Lin is supported by NSFC U2241211, NSFC U20B2046, and 23H020101910. 
% Ying Zhang is supported by ARC DP210101393, LP210301046, and DP230101445. 
% Kai Wang is the corresponding author.

%\clearpage

\bibliographystyle{ACM-Reference-Format}
\bibliography{sample}

\end{document}